\documentclass[dvipsnames]{article}
\usepackage{microtype}
\usepackage{subfigure}
\usepackage{booktabs}
\usepackage{hyperref}

\usepackage{amsmath}
\usepackage[utf8]{inputenc}
\usepackage{multicol}
\usepackage{marvosym}
\usepackage{graphicx}
\graphicspath{ {./images/} }
\usepackage{booktabs} 
\usepackage[T1]{fontenc}
\usepackage{lmodern}
\usepackage{natbib}
\usepackage{cite}
\usepackage{enumerate}
\usepackage{float}
\usepackage{nameref}
\usepackage{varioref}
\usepackage{hyperref}
\usepackage{cleveref}
\usepackage{amssymb}
\usepackage{amsfonts}
\usepackage{amsthm}
\usepackage{textcomp}
\usepackage{gensymb}
\usepackage{bm}
\usepackage{pgfplots}
\usepackage{CJKutf8}
\usepackage{pgf-umlsd}
\newtheorem{theorem}{Theorem}[section]

\newtheorem{proposition}[theorem]{Proposition}

\usepackage{algorithm}
\usepackage{pifont}
\usepackage{multirow,array}
\usepackage{centernot}
\usepackage[inline]{enumitem}
\usepackage{comment}
\usepackage{lipsum}
\usepackage{tcolorbox}
\tcbuselibrary{skins}
\usetikzlibrary{shadings}
\usetikzlibrary{matrix}
\usepackage{color}
\definecolor{deepblue}{rgb}{0,0,0.5}
\definecolor{deepred}{rgb}{0.6,0,0}
\definecolor{deepgreen}{rgb}{0,0.5,0}
\definecolor{asparagus}{rgb}{0.53, 0.66, 0.42}
\definecolor{coolblack}{rgb}{0.0, 0.18, 0.39}
\definecolor{amaranth}{rgb}{0.9, 0.17, 0.31}
\definecolor{amber}{rgb}{1.0, 0.49, 0.0}
\definecolor{blizzardblue}{rgb}{0.67, 0.9, 0.93}
\definecolor{brightpink}{rgb}{1.0, 0.0, 0.5}
\usepackage{listings}
\usepackage{tikz}
\usepackage{tikzsymbols}
\usepackage{skak}
\usetikzlibrary{shapes.geometric,backgrounds,patterns, trees}
\usepackage[linguistics]{forest}
\usetikzlibrary{shapes.geometric,backgrounds,patterns, trees}
\usetikzlibrary{3d,decorations.text,shapes.arrows,positioning,fit,backgrounds}
\tikzset{pics/fake box/.style args={
#1 with dimensions #2 and #3 and #4}{
code={
\draw[gray,ultra thin,fill=#1]  (0,0,0) coordinate(-front-bottom-left) to
++ (0,#3,0) coordinate(-front-top-right) --++
(#2,0,0) coordinate(-front-top-right) --++ (0,-#3,0)
coordinate(-front-bottom-right) -- cycle;
\draw[gray,ultra thin,fill=#1] (0,#3,0)  --++
 (0,0,#4) coordinate(-back-top-left) --++ (#2,0,0)
 coordinate(-back-top-right) --++ (0,0,-#4)  -- cycle;
\draw[gray,ultra thin,fill=#1!80!black] (#2,0,0) --++ (0,0,#4) coordinate(-back-bottom-right)
--++ (0,#3,0) --++ (0,0,-#4) -- cycle;
\path[gray,decorate,decoration={text effects along path,text={CONV}}] (#2/2,{2+(#3-2)/2},0) -- (#2/2,0,0);
}
}}
\tikzset{circle dotted/.style={dash pattern=on .05mm off 2mm,
    line cap=round}}

\usetikzlibrary{shapes.geometric,backgrounds,patterns, trees}
\usetikzlibrary{arrows.meta,
                bending,
                intersections,
                quotes,
                shapes.geometric}
\usetikzlibrary{automata, positioning}
\usetikzlibrary{shapes,arrows}
\usetikzlibrary{arrows.meta}
\usetikzlibrary{automata, positioning}
  \usetikzlibrary{shapes,shadows}
  \tikzstyle{abstractbox} = [draw=black, fill=white, rectangle,
  inner sep=10pt, style=rounded corners, drop shadow={fill=black,
  opacity=1}]
\tikzstyle{abstracttitle} =[fill=white]
\usetikzlibrary{calc,positioning,shapes.geometric}
\usetikzlibrary{arrows.meta,arrows}
\usetikzlibrary{matrix}
\colorlet{myRed}{red!20}
\tikzset{
  rows/.style 2 args={/utils/temp/.style={row ##1/.append style={nodes={#2}}},
    /utils/temp/.list={#1}},
  columns/.style 2 args={/utils/temp/.style={column ##1/.append style={nodes={#2}}},
    /utils/temp/.list={#1}}}
\makeatletter
\tikzset{anchor/.append code=\let\tikz@auto@anchor\relax,
  add font/.code=%
    \expandafter\def\expandafter\tikz@textfont\expandafter{\tikz@textfont#1},
  left delimiter/.style 2 args={append after command={\tikz@delimiter{south east}
    {south west}{every delimiter,every left delimiter,#2}{south}{north}{#1}{.}{\pgf@y}}}}
\tikzstyle{sms} = [rectangle callout, draw,very thick, rounded corners, minimum height=20pt]
\usetikzlibrary{positioning,calc}
\makeatletter
\pgfdeclareshape{document}{
\inheritsavedanchors[from=rectangle] 
\inheritanchorborder[from=rectangle]
\inheritanchor[from=rectangle]{center}
\inheritanchor[from=rectangle]{north}
\inheritanchor[from=rectangle]{south}
\inheritanchor[from=rectangle]{west}
\inheritanchor[from=rectangle]{east}
\backgroundpath{
\southwest \pgf@xa=\pgf@x \pgf@ya=\pgf@y
\northeast \pgf@xb=\pgf@x \pgf@yb=\pgf@y
\pgf@xc=\pgf@xb \advance\pgf@xc by-10pt
\pgf@yc=\pgf@yb \advance\pgf@yc by-10pt
\pgfpathmoveto{\pgfpoint{\pgf@xa}{\pgf@ya}}
\pgfpathlineto{\pgfpoint{\pgf@xa}{\pgf@yb}}
\pgfpathlineto{\pgfpoint{\pgf@xc}{\pgf@yb}}
\pgfpathlineto{\pgfpoint{\pgf@xb}{\pgf@yc}}
\pgfpathlineto{\pgfpoint{\pgf@xb}{\pgf@ya}}
\pgfpathclose
\pgfpathmoveto{\pgfpoint{\pgf@xc}{\pgf@yb}}
\pgfpathlineto{\pgfpoint{\pgf@xc}{\pgf@yc}}
\pgfpathlineto{\pgfpoint{\pgf@xb}{\pgf@yc}}
\pgfpathlineto{\pgfpoint{\pgf@xc}{\pgf@yc}}
}
}
\tikzstyle{doc}=[%
draw,
thick,
align=center,
color=black,
shape=document,
minimum width=20mm,
minimum height=28.2mm,
shape=document,
inner sep=2ex,
]
\tikzstyle{vecArrow} = [thick, decoration={markings,mark=at position
   1 with {\arrow[semithick]{open triangle 60}}},
   double distance=1.4pt, shorten >= 5.5pt,
   preaction = {decorate},
   postaction = {draw,line width=1.4pt, white,shorten >= 4.5pt}]
   \tikzstyle{innerWhite} = [semithick, white,line width=1.4pt, shorten >= 4.5pt]

\tikzset{
  mybackground/.style={execute at end picture={
        \begin{scope}[on background layer]
          \draw[black!15,fill=red!20,rounded corners=1ex] (current bounding box.south west)
                    rectangle (current bounding box.north east);
          \node[draw,fill=white,rectangle, anchor=west,inner sep=1pt,minimum width=4ex] at (current bounding box.north
                   west){#1};
        \end{scope}
    }},
}

\tikzset{
  mybackground2/.style={execute at end picture={
        \begin{scope}[on background layer]
          \draw[black!5,fill=black!15,rounded corners=1ex] (current bounding box.south west)
                    rectangle (current bounding box.north east);
          \node[draw,fill=white,rectangle,anchor=west,inner sep=4pt,minimum width=4ex] at (current bounding box.north
                   west){#1};
        \end{scope}
    }},
}

\tikzset{
  mybackground3/.style={execute at end picture={
        \begin{scope}[on background layer]
          \draw[black!5,fill=black!80, opacity=1,rounded corners=1ex] (current bounding box.south west)
                    rectangle (current bounding box.north east);
          \node[draw,fill=white,ellipse,anchor=west,inner sep=1pt,minimum width=4ex] at (current bounding box.north
                   west){#1};
        \end{scope}
    }},
}

\tikzset{
  mybackground4/.style={execute at end picture={
        \begin{scope}[on background layer]
          \draw[black!15,fill=black!80,rounded corners=1ex] (current bounding box.south west)
                    rectangle (current bounding box.north east);
          \node[draw,fill=white,ellipse,anchor=west,inner sep=1pt,minimum width=4ex] at (current bounding box.north
                   west){#1};
        \end{scope}
    }},
}

\tikzset{
  mybackground5/.style={execute at end picture={
        \begin{scope}[on background layer]
          \draw[black!5,fill=black!80!green!40,rounded corners=1ex] (current bounding box.south west)
                    rectangle (current bounding box.north east);
          \node[draw,fill=white,ellipse,anchor=west,inner sep=1pt,minimum width=4ex] at (current bounding box.north
                   west){#1};
        \end{scope}
    }},
}
\tikzset{
  mybackground6/.style={execute at end picture={
        \begin{scope}[on background layer]
          \draw[black!15,fill=black!15,rounded corners=1ex] (current bounding box.south west)
                    rectangle (current bounding box.north east);
          \node[draw,fill=white,ellipse,anchor=west,inner sep=1pt,minimum width=4ex] at (current bounding box.north
                   west){#1};
        \end{scope}
    }},
}

\tikzset{
  mybackground7/.style={execute at end picture={
        \begin{scope}[on background layer]
          \draw[black!5,fill=black!5,rounded corners=1ex] (current bounding box.south west)
                    rectangle (current bounding box.north east);
          \node[draw,fill=white,ellipse,anchor=west,inner sep=1pt,minimum width=4ex] at (current bounding box.north
                   west){#1};
        \end{scope}
    }},
}
\tikzset{
  mybackground9/.style={execute at end picture={
        \begin{scope}[on background layer]
          \draw[black!1,fill=white,rounded corners=1ex] (current bounding box.south west)
                    rectangle (current bounding box.north east);
          \node[draw,fill=white,ellipse,anchor=west,inner sep=1pt,minimum width=4ex] at (current bounding box.north
                   west){#1};
        \end{scope}
    }},
}

\tikzset{
  mybackground51/.style={execute at end picture={
      \begin{scope}[on background layer]
        \draw[black, rounded corners=2ex, fill=gray2] (current bounding box.south west)
        rectangle (current bounding box.north east);
        \node[draw,fill=white,ellipse,anchor=west,inner sep=1pt,minimum width=1ex] at (current bounding box.north
        west){#1};
      \end{scope}
    }},
}

    \tikzset{
        table/.style={
            matrix of nodes,
            row sep=-\pgflinewidth,
            column sep=-\pgflinewidth,
            nodes={
                rectangle,
                draw=black,
                align=center
            },
            minimum height=1.5em,
            text depth=0.5ex,
            text height=2ex,
            nodes in empty cells,
            every even row/.style={
                nodes={fill=gray!20}
            },
            column 1/.style={
                nodes={text width=5em,font=\bfseries}
            },
            row 1/.style={
                nodes={
                    fill=black,
                    text=white,
                    font=\bfseries
                }
            }
        }
      }
\usetikzlibrary{positioning, fit, arrows.meta, shapes}
\usepackage{bm}
\usepackage{relsize}
\tikzstyle{stateTransition}=[-stealth, thick]
\usetikzlibrary{decorations.pathmorphing, positioning}
\tikzstyle{inputNode}=[draw,circle,minimum size=10pt,inner sep=0pt]
\definecolor{echoreg}{HTML}{2cb1e1}
\definecolor{echodrk}{HTML}{0099cc}

\tikzstyle{mybox} = [text=black, very thick,
    rectangle, rounded corners, inner sep=10pt, inner ysep=20pt]
    \tikzstyle{fancytitle} =[text=black]

\pgfkeys{/pgf/.cd,
  parallelepiped offset x/.initial=2mm,
  parallelepiped offset y/.initial=2mm
}
\pgfplotsset{compat=1.16}
\usepgfplotslibrary{fillbetween}
\pgfdeclareshape{parallelepiped}
{
  \inheritsavedanchors[from=rectangle] 
  \inheritanchorborder[from=rectangle]
  \inheritanchor[from=rectangle]{north}
  \inheritanchor[from=rectangle]{north west}
  \inheritanchor[from=rectangle]{north east}
  \inheritanchor[from=rectangle]{center}
  \inheritanchor[from=rectangle]{west}
  \inheritanchor[from=rectangle]{east}
  \inheritanchor[from=rectangle]{mid}
  \inheritanchor[from=rectangle]{mid west}
  \inheritanchor[from=rectangle]{mid east}
  \inheritanchor[from=rectangle]{base}
  \inheritanchor[from=rectangle]{base west}
  \inheritanchor[from=rectangle]{base east}
  \inheritanchor[from=rectangle]{south}
  \inheritanchor[from=rectangle]{south west}
  \inheritanchor[from=rectangle]{south east}
  \backgroundpath{
    \southwest \pgf@xa=\pgf@x \pgf@ya=\pgf@y
    \northeast \pgf@xb=\pgf@x \pgf@yb=\pgf@y
    \pgfmathsetlength\pgfutil@tempdima{\pgfkeysvalueof{/pgf/parallelepiped
      offset x}}
    \pgfmathsetlength\pgfutil@tempdimb{\pgfkeysvalueof{/pgf/parallelepiped
      offset y}}
    \def\ppd@offset{\pgfpoint{\pgfutil@tempdima}{\pgfutil@tempdimb}}
    \pgfpathmoveto{\pgfqpoint{\pgf@xa}{\pgf@ya}}
    \pgfpathlineto{\pgfqpoint{\pgf@xb}{\pgf@ya}}
    \pgfpathlineto{\pgfqpoint{\pgf@xb}{\pgf@yb}}
    \pgfpathlineto{\pgfqpoint{\pgf@xa}{\pgf@yb}}
    \pgfpathclose
    \pgfpathmoveto{\pgfqpoint{\pgf@xb}{\pgf@ya}}
    \pgfpathlineto{\pgfpointadd{\pgfpoint{\pgf@xb}{\pgf@ya}}{\ppd@offset}}
    \pgfpathlineto{\pgfpointadd{\pgfpoint{\pgf@xb}{\pgf@yb}}{\ppd@offset}}
    \pgfpathlineto{\pgfpointadd{\pgfpoint{\pgf@xa}{\pgf@yb}}{\ppd@offset}}
    \pgfpathlineto{\pgfqpoint{\pgf@xa}{\pgf@yb}}
    \pgfpathmoveto{\pgfqpoint{\pgf@xb}{\pgf@yb}}
    \pgfpathlineto{\pgfpointadd{\pgfpoint{\pgf@xb}{\pgf@yb}}{\ppd@offset}}
  }
}

\makeatletter
\tikzset{anchor/.append code=\let\tikz@auto@anchor\relax,
  add font/.code=%
    \expandafter\def\expandafter\tikz@textfont\expandafter{\tikz@textfont#1},
  left delimiter/.style 2 args={append after command={\tikz@delimiter{south east}
    {south west}{every delimiter,every left delimiter,#2}{south}{north}{#1}{.}{\pgf@y}}}}
\tikzstyle{sms} = [rectangle callout, draw,very thick, rounded corners, minimum height=20pt]
\makeatletter
\tikzset{anchor/.append code=\let\tikz@auto@anchor\relax,
  add font/.code=%
    \expandafter\def\expandafter\tikz@textfont\expandafter{\tikz@textfont#1},
  left delimiter/.style 2 args={append after command={\tikz@delimiter{south east}
    {south west}{every delimiter,every left delimiter,#2}{south}{north}{#1}{.}{\pgf@y}}}}
\tikzstyle{sms} = [rectangle callout, draw,very thick, rounded corners, minimum height=20pt]
\usetikzlibrary{positioning,calc}

\tikzset{l3 switch/.style={
    parallelepiped,fill=switch, draw=white,
    minimum width=0.75cm,
    minimum height=0.75cm,
    parallelepiped offset x=1.75mm,
    parallelepiped offset y=1.25mm,
    path picture={
      \node[fill=white,
        circle,
        minimum size=6pt,
        inner sep=0pt,
        append after command={
          \pgfextra{
            \foreach \angle in {0,45,...,360}
            \draw[-latex,fill=white] (\tikzlastnode.\angle)--++(\angle:2.25mm);
          }
        }
      ]
       at ([xshift=-0.75mm,yshift=-0.5mm]path picture bounding box.center){};
    }
  },
  ports/.style={
    line width=0.3pt,
    top color=gray!20,
    bottom color=gray!80
  },
  rack switch/.style={
    parallelepiped,fill=white, draw,
    minimum width=1.25cm,
    minimum height=0.25cm,
    parallelepiped offset x=2mm,
    parallelepiped offset y=1.25mm,
    xscale=-1,
    path picture={
      \draw[top color=gray!5,bottom color=gray!40]
      (path picture bounding box.south west) rectangle
      (path picture bounding box.north east);
      \coordinate (A-west) at ([xshift=-0.2cm]path picture bounding box.west);
      \coordinate (A-center) at ($(path picture bounding box.center)!0!(path
        picture bounding box.south)$);
      \foreach \x in {0.275,0.525,0.775}{
        \draw[ports]([yshift=-0.05cm]$(A-west)!\x!(A-center)$)
          rectangle +(0.1,0.05);
        \draw[ports]([yshift=-0.125cm]$(A-west)!\x!(A-center)$)
          rectangle +(0.1,0.05);
       }
      \coordinate (A-east) at (path picture bounding box.east);
      \foreach \x in {0.085,0.21,0.335,0.455,0.635,0.755,0.875,1}{
        \draw[ports]([yshift=-0.1125cm]$(A-east)!\x!(A-center)$)
          rectangle +(0.05,0.1);
      }
    }
  },
  server/.style={
    parallelepiped,
    fill=white, draw,
    minimum width=0.35cm,
    minimum height=0.75cm,
    parallelepiped offset x=3mm,
    parallelepiped offset y=2mm,
    xscale=-1,
    path picture={
      \draw[top color=gray!5,bottom color=gray!40]
      (path picture bounding box.south west) rectangle
      (path picture bounding box.north east);
      \coordinate (A-center) at ($(path picture bounding box.center)!0!(path
        picture bounding box.south)$);
      \coordinate (A-west) at ([xshift=-0.575cm]path picture bounding box.west);
      \draw[ports]([yshift=0.1cm]$(A-west)!0!(A-center)$)
        rectangle +(0.2,0.065);
      \draw[ports]([yshift=0.01cm]$(A-west)!0.085!(A-center)$)
        rectangle +(0.15,0.05);
      \fill[black]([yshift=-0.35cm]$(A-west)!-0.1!(A-center)$)
        rectangle +(0.235,0.0175);
      \fill[black]([yshift=-0.385cm]$(A-west)!-0.1!(A-center)$)
        rectangle +(0.235,0.0175);
      \fill[black]([yshift=-0.42cm]$(A-west)!-0.1!(A-center)$)
        rectangle +(0.235,0.0175);
    }
  },
}

\usetikzlibrary{calc, shadings, shadows, shapes.arrows}

\tikzset{%
  interface/.style={draw, rectangle, rounded corners, font=\LARGE\sffamily},
  ethernet/.style={interface, fill=yellow!50},
  serial/.style={interface, fill=green!70},
  speed/.style={sloped, anchor=south, font=\large\sffamily},
  route/.style={draw, shape=single arrow, single arrow head extend=4mm,
    minimum height=1.7cm, minimum width=3mm, white, fill=switch!20,
    drop shadow={opacity=.8, fill=switch}, font=\tiny}
}
\newcommand*{\shift}{1.3cm}
\newcommand\numeq[1]%
{\stackrel{\scriptscriptstyle(\mkern-1.5mu#1\mkern-1.5mu)}{=}}
\newcommand\numeqq[1]%
{\stackrel{\scriptscriptstyle(\mkern-1.5mu#1\mkern-1.5mu)}{\triangleq}}
\newcommand\numleq[1]%
{\stackrel{\scriptscriptstyle(\mkern-1.5mu#1\mkern-1.5mu)}{\leq}}
\newcommand\numgeq[1]%
{\stackrel{\scriptscriptstyle(\mkern-1.5mu#1\mkern-1.5mu)}{\geq}}
\newcommand\numimp[1]%
{\stackrel{\scriptscriptstyle(\mkern-1.5mu#1\mkern-1.5mu)}{\implies}}
\newcommand\norm[1]{\lVert#1\rVert}
\newcommand*{\router}[1]{
\begin{tikzpicture}
  \coordinate (ll) at (-3,0.5);
  \coordinate (lr) at (3,0.5);
  \coordinate (ul) at (-3,2);
  \coordinate (ur) at (3,2);
  \shade [shading angle=90, left color=switch, right color=white] (ll)
    arc (-180:-60:3cm and .75cm) -- +(0,1.5) arc (-60:-180:3cm and .75cm)
    -- cycle;
  \shade [shading angle=270, right color=switch, left color=white!50] (lr)
    arc (0:-60:3cm and .75cm) -- +(0,1.5) arc (-60:0:3cm and .75cm) -- cycle;
  \draw [thick] (ll) arc (-180:0:3cm and .75cm)
    -- (ur) arc (0:-180:3cm and .75cm) -- cycle;
  \draw [thick, shade, upper left=switch, lower left=switch,
    upper right=switch, lower right=white] (ul)
    arc (-180:180:3cm and .75cm);
  \node at (0,0.5){\color{blue!60!black}\Huge #1};
  \begin{scope}[yshift=2cm, yscale=0.28, transform shape]
    \node[route, rotate=45, xshift=\shift] {\strut};
    \node[route, rotate=-45, xshift=-\shift] {\strut};
    \node[route, rotate=-135, xshift=\shift] {\strut};
    \node[route, rotate=135, xshift=-\shift] {\strut};
  \end{scope}
\end{tikzpicture}}

\definecolor{switch}{HTML}{006996}

\tikzset{l3 switch/.style={
    parallelepiped,fill=switch, draw=white,
    minimum width=0.75cm,
    minimum height=0.75cm,
    parallelepiped offset x=1.75mm,
    parallelepiped offset y=1.25mm,
    path picture={
      \node[fill=white,
      circle,
      minimum size=6pt,
      inner sep=0pt,
      append after command={
        \pgfextra{
          \foreach \angle in {0,45,...,360}
          \draw[-latex,fill=white] (\tikzlastnode.\angle)--++(\angle:2.25mm);
        }
      }
      ]
      at ([xshift=-0.75mm,yshift=-0.5mm]path picture bounding box.center){};
    }
  },
  ports/.style={
    line width=0.3pt,
    top color=gray!20,
    bottom color=gray!80
  },
  rack switch/.style={
    parallelepiped,fill=white, draw,
    minimum width=1.25cm,
    minimum height=0.25cm,
    parallelepiped offset x=2mm,
    parallelepiped offset y=1.25mm,
    xscale=-1,
    path picture={
      \draw[top color=gray!5,bottom color=gray!40]
      (path picture bounding box.south west) rectangle
      (path picture bounding box.north east);
      \coordinate (A-west) at ([xshift=-0.2cm]path picture bounding box.west);
      \coordinate (A-center) at ($(path picture bounding box.center)!0!(path
      picture bounding box.south)$);
      \foreach \x in {0.275,0.525,0.775}{
        \draw[ports]([yshift=-0.05cm]$(A-west)!\x!(A-center)$)
        rectangle +(0.1,0.05);
        \draw[ports]([yshift=-0.125cm]$(A-west)!\x!(A-center)$)
        rectangle +(0.1,0.05);
      }
      \coordinate (A-east) at (path picture bounding box.east);
      \foreach \x in {0.085,0.21,0.335,0.455,0.635,0.755,0.875,1}{
        \draw[ports]([yshift=-0.1125cm]$(A-east)!\x!(A-center)$)
        rectangle +(0.05,0.1);
      }
    }
  },
  server/.style={
    parallelepiped,
    fill=white, draw,
    minimum width=0.35cm,
    minimum height=0.75cm,
    parallelepiped offset x=3mm,
    parallelepiped offset y=2mm,
    xscale=-1,
    path picture={
      \draw[top color=gray!5,bottom color=gray!40]
      (path picture bounding box.south west) rectangle
      (path picture bounding box.north east);
      \coordinate (A-center) at ($(path picture bounding box.center)!0!(path
      picture bounding box.south)$);
      \coordinate (A-west) at ([xshift=-0.575cm]path picture bounding box.west);
      \draw[ports]([yshift=0.1cm]$(A-west)!0!(A-center)$)
      rectangle +(0.2,0.065);
      \draw[ports]([yshift=0.01cm]$(A-west)!0.085!(A-center)$)
      rectangle +(0.15,0.05);
      \fill[black]([yshift=-0.35cm]$(A-west)!-0.1!(A-center)$)
      rectangle +(0.235,0.0175);
      \fill[black]([yshift=-0.385cm]$(A-west)!-0.1!(A-center)$)
      rectangle +(0.235,0.0175);
      \fill[black]([yshift=-0.42cm]$(A-west)!-0.1!(A-center)$)
      rectangle +(0.235,0.0175);
    }
  },
}

\makeatletter
\pgfdeclareradialshading[tikz@ball]{cloud}{\pgfpoint{-0.275cm}{0.4cm}}{%
  color(0cm)=(tikz@ball!75!white);
  color(0.1cm)=(tikz@ball!85!white);
  color(0.2cm)=(tikz@ball!95!white);
  color(0.7cm)=(tikz@ball!89!black);
  color(1cm)=(tikz@ball!75!black)
}
\tikzoption{cloud color}{\pgfutil@colorlet{tikz@ball}{#1}%
  \def\tikz@shading{cloud}\tikz@addmode{\tikz@mode@shadetrue}}
\makeatother

\tikzset{my cloud/.style={
     cloud, draw, aspect=2,
     cloud color={gray!5!white}
  }
}
\definecolor{gray2}{HTML}{ededed}
\definecolor{gray3}{HTML}{F5F5F5}
\definecolor{RoyalAzure}{rgb}{0.0, 0.22, 0.66}

\tikzset{
  mybackground18/.style={execute at end picture={
      \begin{scope}[on background layer]
        \draw[black, fill=gray3, rounded corners=3.5ex] (current bounding box.south west)
        rectangle (current bounding box.north east);
        \node[draw,fill=white,ellipse,anchor=west,inner sep=1pt,minimum width=4ex] at (current bounding box.north
        west){#1};
      \end{scope}
    }}
}
\usepackage{pgfplotstable}
\usetikzlibrary{spy}

\tikzset{
  mybackground11/.style={execute at end picture={
        \begin{scope}[on background layer]
          \draw[black, fill=Black!80!Sepia!9, rounded corners=6ex] (current bounding box.south west)
                    rectangle (current bounding box.north east);
          \node[draw,fill=white,ellipse,anchor=west,inner sep=1pt,minimum width=4ex] at (current bounding box.north
                   west){#1};
        \end{scope}
    }},
}

\tikzset{
  mybackground13/.style={execute at end picture={
        \begin{scope}[on background layer]
          \draw[black, fill=gray2, rounded corners=4ex] (current bounding box.south west)
                    rectangle (current bounding box.north east);
          \node[draw,fill=white,ellipse,anchor=west,inner sep=1pt,minimum width=4ex] at (current bounding box.north
                   west){#1};
        \end{scope}
    }},
}

\tikzset{
  mybackground14/.style={execute at end picture={
        \begin{scope}[on background layer]
          \draw[black, rounded corners=2ex] (current bounding box.south west)
                    rectangle (current bounding box.north east);
          \node[draw,fill=white,ellipse,anchor=west,inner sep=1pt,minimum width=4ex] at (current bounding box.north
                   west){#1};
        \end{scope}
    }},
}

\definecolor{lightgray}{gray}{0.9}
\definecolor{lightgreen}{rgb}{0.88, 1, 0.88}
\definecolor{lightred}{rgb}{1, 0.88, 0.88}
\definecolor{bluetwo}{RGB}{189, 213, 234}
\definecolor{bluethree}{RGB}{165, 193, 224}
\definecolor{bluefour}{RGB}{141, 169, 200}

\newcommand{\cmark}{\textcolor{green}{\ding{51}}}
\newcommand{\xmark}{\textcolor{red}{\ding{55}}}
\newcommand{\qmark}{\textcolor{Blue}{\textbf{?}}}
\usepackage{colortbl}

\lstdefinestyle{pythonstyle}{
  basicstyle=\scriptsize\ttfamily,
  language=Python,
  keywordstyle=\color{blue}\bfseries,
  commentstyle=\color{gray}\itshape,
  stringstyle=\color{Red},
  showstringspaces=false,
  frame=single,
  rulecolor=\color{black!30},
  breaklines=true,
  tabsize=2,
  captionpos=b
}

\usepackage[accepted]{mlsys2025}
\mlsystitlerunning{The Cyber Security Learning Environment (CSLE)}
\allowdisplaybreaks

\usepackage{eso-pic}
\AddToShipoutPictureBG*{%
  \put(\LenToUnit{\paperwidth-75pt},\LenToUnit{\paperheight-65pt}){%
    \href{https://www.acm.org/publications/policies/artifact-review-and-badging-current}{%
      \includegraphics[width=65pt]{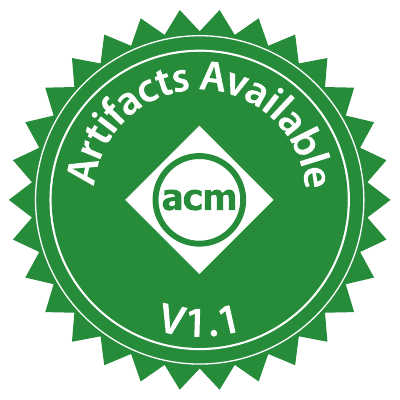}%
    }%
  }%
}

\begin{document}
\twocolumn[
\mlsystitle{CSLE: A Reinforcement Learning Platform \\for Autonomous Security Management}
\begin{mlsysauthorlist}
\mlsysauthor{Kim Hammar}{to}
\end{mlsysauthorlist}

\mlsysaffiliation{to}{Division of Network and Systems Engineering, KTH Royal Institute of Technology, Sweden}
\mlsyscorrespondingauthor{Kim Hammar}{kimham@kth.se}
\mlsyskeywords{Machine Learning, MLSys}

\vskip 0.3in
\begin{abstract}
Reinforcement learning is a promising approach to autonomous and adaptive security management in networked systems. However, current reinforcement learning solutions for security management are mostly limited to simulation environments and it is unclear how they generalize to operational systems. In this paper, we address this limitation by presenting CSLE: a reinforcement learning platform for autonomous security management that enables experimentation under realistic conditions. Conceptually, CSLE encompasses two systems. First, it includes an emulation system that replicates key components of the target system in a virtualized environment. We use this system to gather measurements and logs, based on which we identify a system model, such as a Markov decision process. Second, it includes a simulation system where security strategies are efficiently learned through simulations of the system model. The learned strategies are then evaluated and refined in the emulation system to close the gap between theoretical and operational performance. We demonstrate CSLE through four use cases: flow control, replication control, segmentation control, and recovery control. Through these use cases, we show that CSLE enables near-optimal security management in an environment that approximates an operational system.
\end{abstract}
]

\printAffiliationsAndNotice{}

\section{Introduction}\label{sec:intro}
Managing the security of networked systems alongside their service requirements and physical infrastructures is a major technical challenge that has grown exponentially with the rise of cloud computing, distributed networks, and IoT services. Examples of security management tasks include incident response, risk analysis, strategy design, and threat hunting. Today, many of these tasks remain manual processes carried out by security experts. Although this approach can be effective, it is labor-intensive and requires significant skills. For example, a recent study reports a global shortage of more than $4$ million security experts \cite{ISC2_2024_Workforce_Study}. 

A promising approach to address this challenge is to use reinforcement learning to automatically derive effective security strategies. For example, \citet{li2024conjectural} use reinforcement learning to compute effective incident response strategies. Similarly, \citet{kiely2025cage} use multi-agent reinforcement learning to derive effective defense strategies against advanced persistent threats. A comprehensive review of these developments is provided by \citet{deep_rl_cyber_sec}. While these works report encouraging results, key challenges remain. Chief among them is narrowing the gap between the environment where strategies are evaluated and a scenario playing out in an operational system. Most of the results obtained so far are limited to simulation environments, leaving their practical utility unproven.
\begin{figure}
  \centering
  \scalebox{0.86}{
    \input{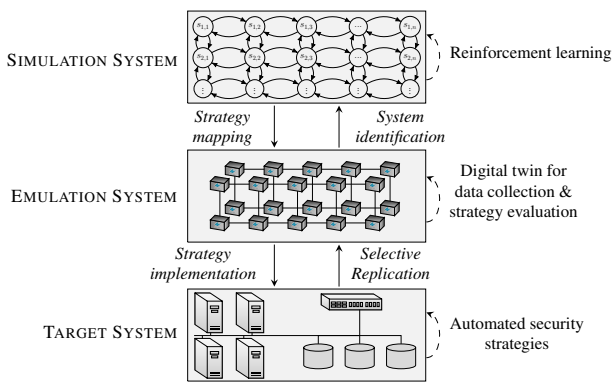}
 }
\caption{Architectural overview of CSLE: a reinforcement learning platform for autonomous security management.}
\label{fig:method}
\end{figure}

In this paper, we address this limitation by presenting a platform that enables experimentation with reinforcement learning in realistic conditions. Conceptually, the platform consists of two systems, as illustrated in Fig.~\ref{fig:method}. First, we use an \textit{emulation system} for creating a virtual replica (i.e., a \textit{digital twin}) of the target system. This twin closely approximates the functionality and timing behavior of the target system, which allows us to run attack scenarios in a controlled environment. Such runs produce system measurements and logs, based on which we identify a system model, e.g., a Markov decision process. Second, we use a \textit{simulation system} where near-optimal security strategies are incrementally learned through reinforcement learning. Learned strategies are extracted from the simulation system and evaluated in the digital twin. This process can be performed iteratively to provide progressively better security strategies that are adapted to changes in the target system, such as configuration changes and software updates.

We refer to the platform as CSLE, which stands for the ``\textit{\textbf{C}yber \textbf{S}ecurity \textbf{L}earning \textbf{E}nvironment}.'' CSLE includes an initial set of 15 digital twin configurations, more than 50 simulated security scenarios, 34 implemented reinforcement learning algorithms, and 4 implemented system identification algorithms, all of which can be extended. To evaluate CSLE experimentally, we use it to learn effective security strategies for four security management tasks: flow control, replication control, segmentation control, and recovery control. Through these use cases, we show that CSLE enables autonomous security management in an environment that closely approximates an operational system. Moreover, we demonstrate the broad applicability of CSLE and its integration with various reinforcement learning techniques.

In summary, the main contributions of this paper are:
\begin{itemize}
\item We present CSLE, a reinforcement learning platform for autonomous security management that enables experimentation under realistic operating conditions.
\item We evaluate CSLE on four different use cases: flow control, replication control, segmentation control, and recovery control. Our experimental results show that CSLE enables autonomous security management.
\end{itemize}  
\paragraph{Open source} The source code of CSLE is released under the CC-BY-SA 4.0 license and available in the repository at \cite{csle_docs}. In addition to the source code, this repository also includes video demonstrations, Docker images, documentation, and datasets of system traces.
\begin{table*}
  \centering
  \scalebox{0.82}{
\begin{tabular}{lcccccccc}
\toprule
\rowcolor{lightgray}
\textbf{Platform} & \textbf{Simulation} & \textbf{Emulation} & \textbf{Open source} & \textbf{RL Library} & \textbf{Management} & \textbf{Validated} & \textbf{Maintained} & \textbf{Distributed} \\
\midrule
\rowcolor{lightgreen}
CSLE (our platform) & \cmark & \cmark & \cmark & \cmark & \cmark & \cmark & \cmark & \cmark\\
CyberBattleSim & \cmark & \xmark & \cmark & \cmark  & \xmark & \xmark & \cmark & \xmark \\
CyBorg & \cmark & \xmark & \cmark & \xmark & \xmark & \xmark & \xmark & \xmark \\
Yawning Titan & \cmark & \xmark & \cmark & \xmark & \xmark & \xmark & \xmark & \xmark \\
NaSim & \cmark & \xmark & \cmark & \xmark & \xmark & \xmark & \xmark & \xmark \\
ATMoS & \xmark & \cmark & \cmark & \xmark & \xmark & \cmark & \xmark & \xmark \\
Gym-FlipIt & \cmark & \xmark & \cmark & \xmark & \xmark & \xmark & \xmark & \xmark \\
Gym-IDSgame & \cmark & \xmark & \cmark & \xmark & \xmark & \xmark & \xmark & \xmark \\
MAB-Malware & \xmark & \cmark & \cmark & \xmark & \xmark & \cmark & \cmark & \xmark \\
Malware-RL & \xmark & \cmark & \cmark & \xmark & \xmark & \cmark & \cmark & \xmark \\
PenGym & \cmark & \cmark & \cmark & \xmark & \xmark & \xmark & \cmark & \xmark \\
CyGil & \xmark & \cmark & \xmark & \xmark & \xmark & \xmark & \qmark & \qmark \\
NaSimEmu & \cmark & \cmark & \cmark & \xmark & \xmark & \xmark & \cmark & \xmark \\
Farland & \cmark & \cmark & \xmark & \xmark & \xmark & \cmark & \qmark & \qmark \\
CyberWheel & \cmark & \cmark & \cmark & \xmark & \xmark & \cmark & \cmark & \xmark \\
CyberShield & \cmark & \xmark & \xmark & \xmark & \xmark & \xmark & \cmark & \xmark \\
Cyborg++ & \cmark & \xmark & \cmark & \xmark & \xmark & \xmark & \cmark & \xmark \\
CyGym & \cmark & \xmark & \cmark & \xmark & \xmark & \xmark & \cmark & \xmark \\
C-CyberBattleSim & \cmark & \xmark & \cmark & \xmark & \xmark & \xmark & \cmark & \xmark \\   
\bottomrule
\end{tabular}
}
\caption{Comparison between reinforcement learning platforms for autonomous security management based on key features: support for simulation-based optimization; support for emulation-based evaluation; open source code; whether the platform provides a library with implemented reinforcement learning algorithms to facilitate strategy optimization and system identification; whether the platform provides a management system for automating experiments and debugging strategies; whether the platform has been experimentally validated on practical use cases; whether the platform is actively maintained; and whether the platform supports distributed deployment.}\label{tab:platform_comparisons}
\end{table*}
\section{Autonomous Security Management through Reinforcement Learning}\label{sec:formalization}
Before presenting our platform, we start by formulating security management as a reinforcement learning problem. To accomplish this formulation, we need a vocabulary in which to talk about the systems and actors involved. To this end, we refer to the operator of the target system as the \textit{defender}, and we refer to an entity aiming to attack the system as the \textit{attacker}. Both interact with the system by taking \textit{actions} (e.g., attacks and responses), which affect the system's \textit{state} (e.g., the system's security and service status). When selecting these actions, the defender and the attacker use measurements from the system (e.g., log files and security alerts), which we refer to as \textit{observations}. A function that maps a sequence of observations to an action is called a \textit{strategy}, and a strategy that is most advantageous according to some objective is \textit{optimal}.

\begin{figure}[H]
  \centering
  \scalebox{0.79}{
    \input{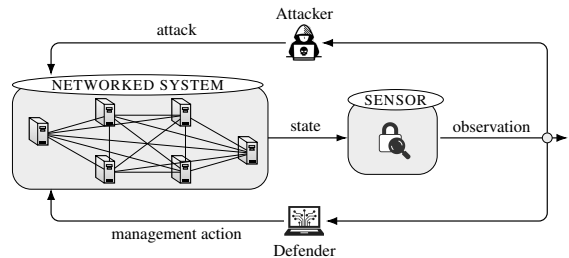}
 }
\caption{Autonomous and adaptive security management of a networked system as a reinforcement learning problem.}
\label{fig:problem_form}
\end{figure}

In the context of reinforcement learning, we can view security management as the problem of learning an effective defender strategy through repeated interaction with the system. In particular, by observing how the system responds to different actions, the defender can gradually improve its strategy to meet security objectives. However, many system-level challenges are encountered when applying this approach in practice. Chief among them are:
\begin{enumerate}
\item In an operational networked system, attacks and defender actions unfold over long time scales and can disrupt critical services. These factors make direct interaction with the system impractical for reinforcement learning, which typically requires executing thousands of actions to learn an effective strategy.  
\item To enable safe and efficient learning, the learning process must, therefore, be executed in a simulation environment. However, the system’s behavior is often too complex to model, which means that the simulation dynamics must be estimated from system measurements.
\item After learning a security strategy through simulation, it must be experimentally validated. In particular, the validity of the simulation must be verified by evaluating the learned strategy in an environment that closely approximates the target system.
\end{enumerate}

In the next section, we review existing platforms that attempt to address these challenges and explain their limitations. We then introduce our platform (CSLE), which is designed to overcome these limitations and enable reinforcement learning experimentation in realistic operating conditions.

\section{Reinforcement Learning Platforms for Autonomous Security Management}
Over the past $5$ years, several reinforcement learning environments for autonomous security management have been developed. They include CyberBattleSim by Microsoft \cite{microsoft_red_teaming}, CyBorg by the Australian department of defense \cite{cyborg}, NaSim by the University of Queensland \cite{schwartz_2020}, Yawning Titan by the UK defense science and technology laboratory \cite{causal_neil_agent}, CyGil by Canada's department of defense \cite{li2021cygil}, NaSimEmu by the Czech Technical University in Prague \cite{janisch2023nasimemu}, ATMoS by the University of Waterloo \cite{atmos}, Gym-FlipIt by Northeastern University \cite{qflip}, Gym-IDSgame by KTH Royal Institute of Technology \cite{hammar_stadler_cnsm_20}, MAB-Malware by the University of California (Riverside) \cite{mab_malware}, Malware-RL by the University of Virginia \cite{anderson2018learningevadestaticpe}, PenGym by Japan's advanced institute of science and technology \cite{pengym}, Farland by USA's national security agency \cite{farland}, CyberWheel by the Oak Ridge national laboratory \cite{cyberwheel}, CyberShield by the University of Malaga \cite{cybershield}, Cyborg++ by the Alan Turing Institute \cite{emerson2024cyborgenhancedgymdevelopment}, CyGym by Washington University \cite{lanier2025cygymsimulationbasedgametheoreticanalysis}, and C-CyberBattleSim by the University of Lorraine \cite{terranova:hal-05182437}.

Like CSLE, all of the referenced platforms include capabilities for learning security strategies using reinforcement learning. However, they differ from CSLE in several important ways, as highlighted in Table~\ref{tab:platform_comparisons}. First, most existing platforms are confined to simulations. By contrast, CSLE is centered around an emulation system based on virtualization. The benefit of our approach is that it narrows the gap between the environment where security strategies are evaluated and a scenario playing out in an operational system. Second, many of the referenced platforms are not open source and most of them are no longer maintained. By contrast, CSLE is open source and has an active development community. Third, CSLE has been experimentally validated on a range of practical use cases, whereas most other platforms have only been evaluated on a single simulation use case. Fourth, unlike the other platforms, CSLE supports distributed deployment, which improves scalability. Moreover, CSLE incorporates a novel management system that provides infrastructure for automating reinforcement learning experiments and debugging the learned strategies.

Lastly, we note that a few platforms for autonomous system operations based on large language models (LLMs) have recently emerged, most notably ITBench~\cite{itbench} and AIOpsLab~\cite{aiopslab}. These platforms focus on using LLMs to automate general system operation tasks. By contrast, CSLE is explicitly designed for automating security management tasks. Another difference is that the referenced platforms are designed for evaluating LLM-based agents, whereas CSLE is designed for developing reinforcement learning agents. This difference in scope leads to fundamental differences in platform architecture. In particular, LLMs do not require the same system-level support as reinforcement learning agents do. For example, CSLE supports identifying simulation models, optimizing strategies through reinforcement learning, and transferring strategies from simulation to emulation. None of these functions is provided by ITBench and AIOpsLab.

\section{Architecture of CSLE}\label{sec:end_to_end}
The architecture of CSLE is illustrated in Fig.~\ref{fig:method} and is centered around an emulation system for creating a digital twin, i.e., a virtual replica of the target system.\footnote{By \textit{target system}, we mean the system where the learned security strategies are intended to be deployed.} We use this twin to run automated attack scenarios and defender responses. Such runs produce system measurements and logs, from which we estimate infrastructure statistics. These statistics allow us to instantiate a mathematical model of the target system through \textit{system identification}. We then leverage this model to learn effective security strategies through simulation, whose performance is assessed using the digital twin. This closed-loop process can be executed iteratively to provide progressively better security strategies that are adapted to changes in the target system; see Fig.~\ref{fig:digital_twin}.
\begin{figure}[H]
  \centering
  \scalebox{1}{
    \input{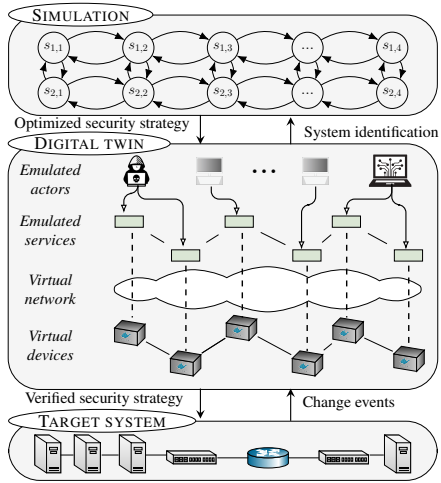}
 }
\caption{A digital twin in CSLE is a virtual replica of a target system that runs the same software and configuration, but on virtualized hardware. Moreover, the twin controls network delays and emulates actors to replicate operational workloads. The twin is used in CSLE for strategy evaluation and system identification.}
\label{fig:digital_twin}
\end{figure}

\subsection{Emulation System}
As described above, the emulation system in CSLE is used to create a digital twin of the target system. The concept of a \textit{digital twin} emerged in the 1960s when NASA used virtual environments to evaluate failure scenarios for lunar landers \cite{nasa_dt}. Since then, digital twin has emerged as a key technology in automation and has been adopted in several industries, including the manufacturing industry [see e.g., \cite{digital_twin_industry}], the automotive industry [see e.g., \cite{digital_twin_automotive}], the healthcare industry [see e.g., \cite{digital_twin_healthcare}], and the technology industry [see e.g., \cite{9429703}].

In CSLE, a digital twin is a virtual replica of a networked system that provides a controlled environment for virtual operations (e.g., cyberattacks and responses), the outcomes of which can be used to optimize operations in the target system. Such a twin enables us to systematically test security strategies under different conditions, including varying attacks, workloads, and network latencies.

\subsection{Simulation System}
The simulation system in CSLE is used to run simulations and execute reinforcement learning algorithms. Although these algorithms could in principle be executed in the digital twin, this approach is not practical due to the long execution times required for carrying out actions and collecting observations in the digital twin. For instance, executing a cyberattack or a defensive reconfiguration in a digital twin can take several minutes. In contrast, the simulation system abstracts these processes as actions in a Markov decision process, which reduces the execution time to milliseconds.

With a \textit{simulation}, we mean an execution of a \textit{discrete-time dynamical system} of the form
\begin{align}
s_{t+1} &\sim f(s_t, a^{(\mathrm{D})}_t, a_t^{(\mathrm{A})}),\label{eq:dt_system}
\end{align}
where $s_t$ is the system state at time $t$, $a^{(\mathrm{D})}_t$ is the defender action, $a^{(\mathrm{A})}_t$ is the attacker action, $f$ is the system dynamics, and $s \sim f$ means that $s$ is sampled from $f$. For example, the dynamics $f$ may represent a Markov decision process (MDP) or a Markov game. Each simulation path $s_1,s_2,\hdots,s_t$ is associated with security consequences and costs. The goal of reinforcement learning is to identify the defender actions that control the simulation in an optimal manner according to a specified security objective.

\subsection{Reinforcement Learning Methodology}
The emulation and simulation systems in CSLE enable a reinforcement learning methodology with the following steps.
\begin{enumerate}[leftmargin=1cm]
\item[\textbf{Step 1}] \textit{Defining the target system}.
\begin{itemize}
\item This is the system where the learned security strategies are intended to be deployed. In CSLE, the target system is defined through a configuration file that specifies the system components, the network topology, the services, etc.
\end{itemize}
\item[\textbf{Step 2}] \textit{Creating a digital twin of the target system.}
\begin{itemize}
\item Given the target system specification, the creation of a digital twin in CSLE is automated through the emulation system.
\end{itemize}
\item[\textbf{Step 3}] \textit{Collecting data from the digital twin.}
\begin{itemize}
\item After creating the digital twin, we use it to run attack scenarios. Such runs produce system traces (i.e., sequences of system metrics), which we collect through CSLE's monitoring system.
\end{itemize}
\item[\textbf{Step 4}] \textit{Identifying a system model.}
\begin{itemize}
\item Having collected system measurements from the digital twin, we use the collected data to identify a model (e.g., through statistical learning) that can be used for running simulations, such as a Markov decision process (MDP).
\end{itemize}
\item[\textbf{Step 5}] \textit{Learning an effective security strategy.}
\begin{itemize}
\item Given the identified system model (e.g., an MDP), we apply reinforcement learning techniques to learn an effective security strategy.

\end{itemize}
\item[\textbf{Step 6}] \textit{Evaluating the learned strategy in the digital twin.}
\begin{itemize}
\item After the learning process has converged, we evaluate the learned security strategy in the digital twin. Such evaluation involves measuring system metrics from the digital twin in real time (e.g., security alerts), using them as input to the security strategy, and executing the action prescribed by the strategy in the digital twin.
\end{itemize}
\item[\textbf{Step 7}] \textit{Deploying the learned strategy in the target system.}
\begin{itemize}
\item If the evaluation is satisfactory, we deploy the learned strategy in the target system. Otherwise, we collect more data to update the simulation and then learn a new strategy. This procedure of updating the simulation and re-learning the strategy is repeated until a strategy with satisfactory performance is obtained.
\end{itemize}  
\end{enumerate}

\section{Implementation of CSLE}\label{sec:architecture}
We have implemented CSLE in Python [$\approx 275,000$ lines of code], JavaScript [$\approx 40,000$ lines of code], and Bash [$\approx 5,000$ lines of code]. From an architectural point of view, the implementation can be divided into three systems: the emulation system, the simulation system, and the management system; see Fig.~\ref{fig:csle_arch}. Broadly speaking, the emulation system creates digital twins, the simulation system runs reinforcement learning algorithms, and the management system orchestrates the platform. The rest of this section delves into the technical details of these three systems.

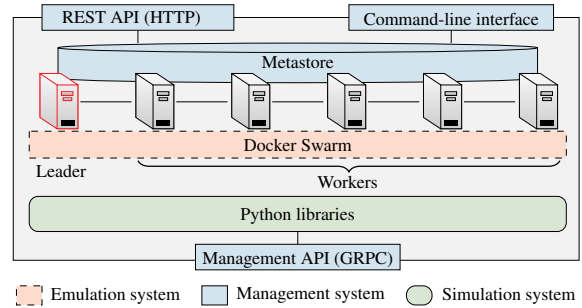
\begin{figure}[H]
  \centering
  \scalebox{0.85}{
          \begin{tikzpicture}[fill=white, >=stealth,
    node distance=3cm,
    database/.style={
      cylinder,
      cylinder uses custom fill,
      shape border rotate=90,
      aspect=0.25,
      draw}]

    \tikzset{
node distance = 9em and 4em,
sloped,
   box/.style = {%
    shape=rectangle,
    rounded corners,
    draw=blue!40,
    fill=blue!15,
    align=center,
    font=\fontsize{12}{12}\selectfont},
 arrow/.style = {%
    line width=0.1mm,
    -{Triangle[length=5mm,width=2mm]},
    shorten >=1mm, shorten <=1mm,
    font=\fontsize{8}{8}\selectfont},
}

\node[scale=1] (management_system) at (0,0)
{
  \begin{tikzpicture}

\draw[fill=black!5] (-0.85, -2.2) rectangle (8, 1.6) {};
\node[database, minimum width=7.5cm, minimum height=0.62cm,fill=MidnightBlue!12,align=center] (s6) at (3.6,0.84) {$\quad$};
\node[server, scale=0.9, color=Red](s1) at (0,0.2) {};
\node[server, scale=0.9](s2) at (1.5,0.2) {};
\node[server, scale=0.9](s3) at (3,0.2) {};
\node[server, scale=0.9](s4) at (4.5,0.2) {};
\node[server, scale=0.9](s5) at (6,0.2) {};
\node[server, scale=0.9](s6) at (7.5,0.2) {};

\node[inner sep=0pt,align=center, scale=0.8, color=black] (hacker) at (-0.1,-0.8) {Leader};

\draw[-, color=black] (0.2,0.28) to (1,0.28);
\draw[-, color=black] (1.7,0.28) to (2.5,0.28);
\draw[-, color=black] (3.2,0.28) to (4,0.28);
\draw[-, color=black] (4.7,0.28) to (5.5,0.28);
\draw[-, color=black] (6.2,0.28) to (7,0.28);
\node[inner sep=0pt,align=center, scale=0.8, color=black] (hacker) at (3.6,0.9) {Metastore};

\draw[fill=MidnightBlue!12] (2, -1.95) rectangle (5.2, -2.35) {};
\draw[fill=MidnightBlue!12] (-0.4, 1.4) rectangle (2.6, 1.8) {};
\draw[fill=MidnightBlue!12] (4.4, 1.4) rectangle (7.6, 1.8) {};

\draw[fill=Red!15, dashed] (-0.6, -0.18) rectangle (7.8, -0.6) {};

\draw[rounded corners, fill=OliveGreen!15] (-0.6, -1.2) rectangle (7.8, -1.75) {};

\node[inner sep=0pt,align=center, scale=0.8, color=black] (hacker) at (3.6,-1.5) {Python libraries};

\node[inner sep=0pt,align=center, scale=0.8, color=black] (hacker) at (3.6,-2.2) {Management API (GRPC)};
\node[inner sep=0pt,align=center, scale=0.8, color=black] (hacker) at (1,1.58) {REST API (HTTP)};
\node[inner sep=0pt,align=center, scale=0.8, color=black] (hacker) at (6,1.58) {Command-line interface};

\draw[-, color=black] (1.1,1.4) to (1.1,1.2);
\draw[-, color=black] (6,1.4) to (6,1.2);
\draw[-, color=black] (3.6,-1.75) to (3.6,-1.95);

\draw [decorate,decoration={brace,amplitude=5pt,mirror,raise=4pt},yshift=0pt,rotate=180, line width=0.20mm]
(-1.1,0.5) -- (-7.7,0.5) node [black,midway,xshift=0.1cm] {};

\node[inner sep=0pt,align=center, scale=0.8, color=black] (hacker) at (4.38,-0.97) {Workers};

\node[inner sep=0pt,align=center, scale=0.8, color=black] (hacker) at (3.6,-0.4) {Docker Swarm}; 

\draw[fill=Red!15, dashed] (-0.4, -2.6) rectangle (-0.8, -2.9) {};

\node[inner sep=0pt,align=center, scale=0.8, color=black] (hacker) at (0.75,-2.75) {Emulation system};

\draw[fill=MidnightBlue!12] (2.1, -2.6) rectangle (2.5, -2.9) {};

\node[inner sep=0pt,align=center, scale=0.8, color=black] (hacker) at (3.8,-2.75) {Management system};

\draw[fill=OliveGreen!15, rounded corners] (5.3, -2.6) rectangle (5.7, -2.9) {};

\node[inner sep=0pt,align=center, scale=0.8, color=black] (hacker) at (6.9,-2.75) {Simulation system};

\end{tikzpicture}
};

\end{tikzpicture}
 }
 \caption{The architecture of CSLE. It is a distributed platform with $N$ servers ($N=6$ in this example), which are connected through a database (the metastore) and a virtualization layer provided by Docker Swarm. CSLE has four interfaces: a Python API, a GRPC API, a REST API, and a command-line interface.}
\label{fig:csle_arch}
\end{figure}

\subsection{Infrastructure}
CSLE runs on a distributed system with $N \geq 1$ servers connected through an IP network. Each server runs a virtualization layer provided by Docker Swarm \cite{docker} and can be accessed through Python libraries, a web interface, a command-line interface, and a GRPC interface \cite{grpc}. Platform metadata is stored in a distributed database referred to as the \textit{metastore}, which is based on Citus \cite{citus}. This database consists of $N$ replicas, one per server. One replica is a designated \textit{leader} and is responsible for coordination. The others are \textit{workers}. A new leader is elected by a quorum whenever the current leader fails or becomes unresponsive. CSLE thus tolerates up to $\lfloor\frac{N-1}{2}\rfloor$ failing servers. This design enables horizontal scaling as the number of servers increases.

Deployment of CSLE in both on-premise and cloud infrastructures is automated using Ansible \cite{ansible}. This automation enables on-demand deployment, allowing CSLE to be launched dynamically for specific experiments or to run continuously as part of an operational environment.

\subsection{The Emulation System}\label{sec:emulation}
The purpose of the emulation system in CSLE is to create a digital twin that replicates relevant components of the target system. Creating such a twin involves three tasks: (\textit{i}) emulating the target system's physical infrastructure, such as processors, network interfaces, and network conditions; (\textit{ii}) emulating actors, i.e., attackers, defenders, and clients; and (\textit{iii}) instrumenting the twin with monitoring and management capabilities. Each of these tasks is detailed below.

\subsubsection*{Emulating hosts and switches}
We emulate hosts and switches with Docker containers \cite{docker}, i.e., lightweight executable packages that include runtime systems, code, libraries, and configurations. This virtualization lets us quickly instantiate large digital twins; see Fig.~\ref{fig:dt_scale}. Resource allocation to containers, e.g., CPU and memory, is enforced using Cgroups. Containers that emulate switches run OVS \cite{ovs} and connect to controllers through OpenFlow \cite{openflow}. Since the switches are programmed through flow tables, they can act as layer-two switches or as routers.

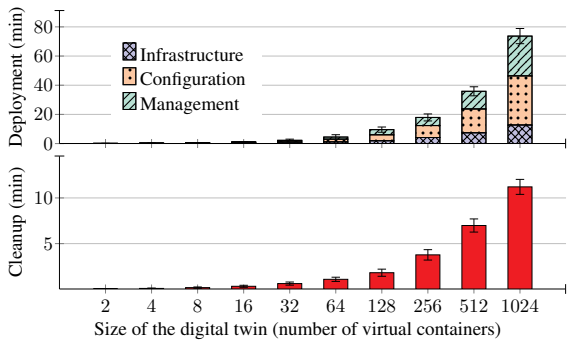
\begin{figure}[H]
  \centering
  \scalebox{0.8}{
    \begin{tikzpicture}
\node[scale=1] (kth_cr) at (0,0)
{
\begin{tikzpicture}    
\pgfplotstableread{
0.056   0.056   0.028 0.01
0.112   0.112   0.056 0.03
0.224   0.224   0.112 0.09
0.448   0.448   0.224 0.2
0.896   0.896   0.448 0.8
1.092   1.892   1.496 1.6
1.884   4.084   3.592 1.8
4.168   8.168   5.584 2.4
7.336  16.336  12.168 3.1
12.672  33.672  27.336 5.2
}\datatablee

\begin{axis}[
    ybar stacked,
    title style={align=center},
    ticks=both,
    xmin=-0.5,
    xmax=9.5,
    ymax=75,
    axis x line=bottom,
    axis y line=left,
    axis line style={-|},
    enlarge y limits={lower, value=0.1},
    enlarge y limits={upper, value=0.22},
    xtick=data,
    ymajorgrids,
    xticklabels={
    },
    y tick label style={align=center, xshift=0.08cm, scale=0.9},                    
    legend style={at={(0.4, 0.85)}, nodes={scale=0.9, transform shape}, anchor=north, legend columns=1,
                draw=none,
      anchor=north east,
      /tikz/every even column/.style={anchor=west, column sep=5pt},
      /tikz/every odd column/.style={anchor=west},          
            /tikz/column 2/.style={
                column sep=5pt,
              }
    },
    every axis legend/.append style={nodes={right}, inner sep=0.2cm},
    x tick label style={align=center, yshift=-0.1cm},
    enlarge x limits=0.05,
    width=10cm,
    height=3.8cm,
    bar width=0.4cm
    ]

\addplot[draw=black, fill=Periwinkle!40, postaction={pattern=crosshatch}] 
    table [x expr=\coordindex, y index=0] {\datatablee};

\addplot[draw=black, fill=Orange!40, postaction={pattern=dots}] 
    table [x expr=\coordindex, y index=1] {\datatablee};

\addplot[draw=black, fill=SeaGreen!40, postaction={pattern=north east lines}, error bars/.cd, y dir=both, y explicit] 
    table [x expr=\coordindex, y index=2, y error index=3] {\datatablee};

\legend{Infrastructure, Configuration, Management}

\end{axis}

\node[inner sep=0pt, align=center, scale=0.9, rotate=90, opacity=1] (obs) at (-0.7, 1.15)
{
  Deployment (min)
};
\end{tikzpicture}
};

\node[scale=1] (kth_cr) at (0,-2.7)
{
\begin{tikzpicture}    
\pgfplotstableread{
0.042 0.01
0.095 0.04
0.184 0.07
0.328 0.11
0.625 0.17
1.092 0.23
1.814 0.39
3.768 0.57
6.982 0.72
11.210 0.82
}\datatablee

\begin{axis}[
    ybar stacked,
    title style={align=center},
    ticks=both,
    xmin=-0.5,
    xmax=9.5,
    ymax=12,
    axis x line=bottom,
    axis y line=left,
    axis line style={-|},
    enlarge y limits={lower, value=0.1},
    enlarge y limits={upper, value=0.22},
    xtick=data,
    ymajorgrids,
    xticklabels={
        $2$,
        $4$,
        $8$,
        $16$,
        $32$,
        $64$,
        $128$,
        $256$,
        $512$,
        $1024$
      },
      x tick label style={align=center, yshift=0.1cm, scale=0.9},
      y tick label style={align=center, xshift=0.08cm, scale=0.9},                      
    legend style={at={(0.4, 0.75)}, nodes={scale=0.9, transform shape}, anchor=north, legend columns=1,
                draw=none,
      anchor=north east,
      /tikz/every even column/.style={anchor=west, column sep=5pt},
      /tikz/every odd column/.style={anchor=west},          
            /tikz/column 2/.style={
                column sep=5pt,
              }
    },
    every axis legend/.append style={nodes={right}, inner sep=0.2cm},
    x tick label style={align=center, yshift=-0.1cm},
    enlarge x limits=0.05,
    width=10cm,
    height=3.8cm,
    bar width=0.4cm
    ]

  \addplot[draw=black, fill=Red,error bars/.cd, y dir=both, y explicit] table [x expr=\coordindex, y index=0,y error index=1] {\datatablee};
\end{axis}

\node[inner sep=0pt, align=center, scale=0.9, rotate=0, opacity=1] (obs) at (4, -0.7)
{
Size of the digital twin (number of virtual containers)
};
\node[inner sep=0pt, align=center, scale=0.9, rotate=90, opacity=1] (obs) at (-0.7, 1.15)
{
  Cleanup (min)
};
\end{tikzpicture}
};

\end{tikzpicture}
 }
\caption{Time to deploy and cleanup a digital twin in CSLE. Deploying the twin involves creating containers, attaching them to networks, configuring them, and starting management services. Cleanup involves stopping and deleting containers and networks. The time measurements were performed for a digital twin with a single network running on a server with a $24$-core Intel Xeon Gold $2.10$ GHz CPU and $768$ GB RAM. Numbers and error bars indicate the mean and the standard deviation from $5$ measurements.}
\label{fig:dt_scale}
\end{figure}

The hosts and switches of the digital twin are specified through a configuration file written in Python, which CSLE parses before deploying the twin. We provide a code snippet of the configuration file in Listing \ref{lst:node_cfg}.
\begin{lstlisting}[caption={Python code for configuring a container.}, label={lst:node_cfg}]
from csle_common.dao.emulation_config.node_container_config import NodeContainerConfig
from csle_common.dao.emulation_config.node_firewall_config import NodeFirewallConfig
node_cfg = NodeContainerConfig(name="my-image", os="Ubuntu22", ips=[..], subnets=[..], interfaces=[..],
       cpus=1, memory_gb=4)
node_fw_config = NodeFirewallConfig(host="..",default_gw="",default_input="ACCEPT",default_output="ACCEPT",default_forward="ACCEPT",fw_rules=[..])
\end{lstlisting}

\subsubsection*{Emulating network links}
We emulate network connectivity in digital twins through virtual links implemented by Linux bridges and network namespaces. If an emulated network spans multiple physical servers, we tunnel the traffic over the physical network using VXLAN \cite{vxlan}. In other words, the physical network of the servers provides a substrate, on top of which the emulated networks are overlaid.

Network conditions of virtual links are created using the NetEm module in the Linux kernel \cite{netem}. This module allows setting bit rates, packet delays, packet loss probabilities, and jitter. For example, the standard configuration in CSLE emulates connections between servers in an IT system with full-duplex, lossless connections of 1 Gbit/s capacity in both directions. Similarly, the default configuration for external communications is full-duplex connections of $100$ Mbit/s capacity and $0.1\%$ packet loss with random bursts of $1\%$ packet loss. These numbers are based on measurements on enterprise and wide-area networks; see e.g., \cite{packet_losses_decreasing,Paxson97end-to-endinternet}.

The network conditions are configured in CSLE through Python objects. We provide an example in Listing \ref{lst:network_cfg}.
\begin{lstlisting}[caption={Python code for configuring a network interface.}, label={lst:network_cfg}]
from csle_common.dao.emulation_config.node_network_config import NodeNetworkConfig  
NodeNetworkConfig(interface="eth0", packet_delay_ms=2, jitter_ms=0.5, delay_distribution="pareto", corrupt=0.02, duplicate=0.00001, correlation=25, reorder=2, rate_limit_mbit=100)
\end{lstlisting}

\subsubsection*{Emulating actors}
All actors in CSLE are programmatically controlled through a management API based on GRPC, which allows changing configuration parameters, starting new actors, and stopping running ones. This automation enables attackers, clients, and defenders to operate in a fully autonomous environment.

We emulate clients through processes in the digital twin that access services on emulated hosts. The client population is defined by (\textit{i}) an arrival process (e.g., a Poisson process) that controls the rate at which new client processes are started; (\textit{ii}) a service time distribution (e.g., an exponential distribution) that controls how long a client will consume services before terminating; (\textit{iii}) a service configuration that specifies the services of the digital twin that clients will consume; and (\textit{iv}) a Markov process that controls the sequence of service invocations that a client makes. All of these parameters are configured in CSLE through a Python file. We provide a code snippet of this file in Listing~\ref{lst:client_cfg}.

\begin{lstlisting}[caption={Python code for configuring the client population.}, label={lst:client_cfg}]
from csle_collector.client_manager.dao.client import Client
from csle_collector.client_manager.dao.constant_arrival_config import ConstantArrivalConfig
clients=[Client(service_distribution=[0.5,0.2,0.3], arrival_config=ConstantArrivalConfig(lamb=20), mu=4, exponential_service_time=True)]
\end{lstlisting}

Figure \ref{fig:cpu_memory_usage} shows the resource usage of two digital twins as a function of the client arrival rate. We observe, as expected, that the resource usage increases with the load imposed on the twins. In particular, higher client arrival rates lead to increased CPU utilization since the twins must process a larger number of service requests. In contrast, the memory usage remains stable when increasing the load.

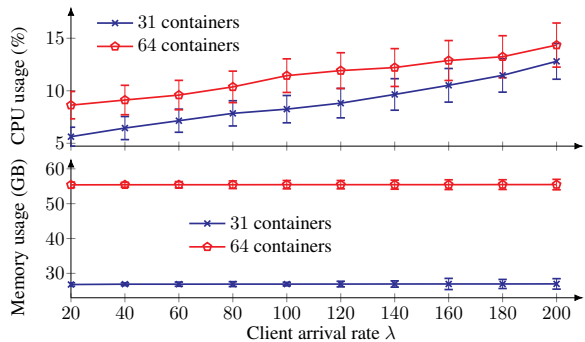
\begin{figure}[H]
  \centering
  \scalebox{0.8}{
    \begin{tikzpicture}
\pgfplotstableread{
20 5.63 26.8 0.9 0.5
40 6.45 26.9 1.1 0.4
60 7.15 26.9 1.1 0.6
80 7.85 26.92 1.2 0.7
100 8.25 26.92 1.3 0.5
120 8.82 26.93 1.4 0.8
140 9.65 26.95 1.5 0.9
160 10.52 26.96 1.6 1.6
180 11.47 26.96 1.6 1.3
200 12.8 26.98 1.7 1.5
}\resourceone

\pgfplotstableread{
20 8.63 55.4 1.3 0.8
40 9.12 55.42 1.4 0.7
60 9.59 55.43 1.4 0.9
80 10.37 55.43 1.5 1.1
100 11.44 55.45 1.6 1.2
120 11.92 55.45 1.7 1.2
140 12.21 55.45 1.8 1.3
160 12.88 55.45 1.9 1.4
180 13.24 55.47 2.0 1.4
200 14.35 55.48 2.1 1.5
}\resourcetwo

\pgfplotsset{/dummy/workaround/.style={/pgfplots/axis on top}}

\node[scale=1] (kth_cr) at (0,0)
{
\begin{tikzpicture}
  \begin{axis}
[
        xmin=20,
        xmax=210,
        width=10.1cm,
        height=3.9cm,
        ymax=18,
        axis y line=center,
        axis x line=bottom,
        scaled y ticks=false,
        xticklabels={},
        yticklabel style={
        /pgf/number format/fixed,
        /pgf/number format/precision=5
      },
        xlabel style={below right},
        ylabel style={above left},
        x tick label style={align=center, yshift=0.04cm, scale=0.9},
        y tick label style={align=center, xshift=0.08cm, scale=0.9},                
        axis line style={-{Latex[length=1.5mm]}},
        legend style={at={(0.35,1.0)}},
        legend columns=1,
        legend style={
          draw=none,
          nodes={scale=0.9, transform shape},
      anchor=north east,
      /tikz/every even column/.style={anchor=west, column sep=5pt},
      /tikz/every odd column/.style={anchor=west},          
            /tikz/column 2/.style={
                column sep=5pt,
              }
            },
            error bars/y dir=both,
error bars/y explicit,
error bars/error mark options={
  rotate=90, 
  mark size=2pt,
  thick
},
              ]
            \addplot[Blue,name path=l1, thick, mark=x, mark repeat=1] table [x index=0, y index=1, y error index=3] {\resourceone};
            \addplot[Red,name path=l1, thick, mark=pentagon, mark repeat=1] table [x index=0, y index=1, y error index=3] {\resourcetwo};
            \legend{$31$ containers, $64$ containers}
            \end{axis}
\node[inner sep=0pt,align=center, scale=0.9, rotate=90, opacity=1] (obs) at (-0.85,1.05)
{
CPU usage (\%)
};
\end{tikzpicture}
};

\node[scale=1] (kth_cr) at (0,-2.7)
{
\begin{tikzpicture}
  \begin{axis}
[
        xmin=20,
        xmax=210,
        width=10.1cm,
        height=3.9cm,
        ymax=63,
        ymin=23,
        axis y line=center,
        axis x line=bottom,
        scaled y ticks=false,
        yticklabel style={
        /pgf/number format/fixed,
        /pgf/number format/precision=5
      },
        xlabel style={below right},
        ylabel style={above left},
        x tick label style={align=center, yshift=0.04cm, scale=0.9},
        y tick label style={align=center, xshift=0.08cm, scale=0.9},                
        axis line style={-{Latex[length=1.5mm]}},
        legend style={at={(0.53,0.65)}},
        legend columns=1,
        legend style={
          draw=none,
          nodes={scale=0.9, transform shape},
      anchor=north east,
      /tikz/every even column/.style={anchor=west, column sep=5pt},
      /tikz/every odd column/.style={anchor=west},          
            /tikz/column 2/.style={
                column sep=5pt,
              }
            },
            error bars/y dir=both,
error bars/y explicit,
error bars/error mark options={
  rotate=90, 
  mark size=2pt,
  thick
},            
              ]
            \addplot[Blue,name path=l1, thick, mark=x, mark repeat=1] table [x index=0, y index=2, y error index=4] {\resourceone};
            \addplot[Red,name path=l1, thick, mark=pentagon, mark repeat=1] table [x index=0, y index=2, y error index=4] {\resourcetwo};
            \legend{$31$ containers, $64$ containers}
            \end{axis}
\node[inner sep=0pt,align=center, scale=0.9, rotate=90, opacity=1] (obs) at (-0.85,1.05)
{
Memory usage (GB)
};
\node[inner sep=0pt,align=center, scale=0.9, rotate=0, opacity=1] (obs) at (4.2,-0.57)
{
  Client arrival rate $\lambda$
};
\end{tikzpicture}
};

\end{tikzpicture}
 }
 \caption{Resource usage of two digital twins in function of the client (Poisson) arrival rate $\lambda$. Numbers and error bars indicate the mean and the standard deviation from $5$ evaluations. The CPU and memory usages are averaged over a monitoring period of $30$ minutes. The blue curves relate to digital twins of an IT infrastructure with $31$ and $64$ hosts, respectively. The network topologies are shown in Figs.~\ref{fig:systems}.a--b and the configurations are available in the supplementary material (Tables~\ref{tab:config1} and \ref{tab:config2}). Each client consumes a randomly selected service of the infrastructure for a time that is sampled from an exponential distribution with mean value $\mu=60$ seconds. We run the digital twins on a server with a $24$-core Intel Xeon Gold $2.10$ GHz CPU and $768$ GB RAM.}
\label{fig:cpu_memory_usage}
\end{figure}

Similar to how clients are emulated, attackers in CSLE are implemented as autonomous processes that execute actions from a pre-defined list, including reconnaissance commands, privilege escalation actions, and exploits. Table~\ref{tab:dt_attacker_actions} lists some of the attacker actions that are automated in CSLE. The defender is emulated in a similar way, with actions implemented as system commands that can reconfigure network components, isolate hosts, or perform other mitigation steps. We provide several examples of defender actions in Table~\ref{tab:defender_stop_actions}.

\begin{table}
  \centering
  \scalebox{0.7}{
    \begin{tabular}{lll} \toprule
\rowcolor{lightgray}      
  {\textit{Type}} & {\textit{Actions}} & {MITRE ATT\&CK technique} \\ \midrule
  Reconnaissance  & TCP SYN scan, UDP scan & T1046 service scanning.\\
                  & TCP XMAS scan & T1046 service scanning. \\
                  & Vulscan & T1595 active scanning. \\
                  & ping-scan & T1018 system discovery.\\\midrule
  Brute-force & Telnet, SSH & T1110 brute force.\\
                  & FTP, Cassandra & T1110 brute force.\\
                  &  IRC, MongoDB, MySQL & T1110 brute force.\\
                  & SMTP, Postgres & T1110 brute force.\\\midrule
  Exploit & CVE-2017-7494 & T1210 service exploitation.\\
                  &CVE-2015-3306 & T1210 service exploitation.\\
                  & CVE-2010-0426 & T1068 privilege escalation.\\
                  & CVE-2015-5602 & T1068 privilege escalation.\\
                  & CVE-2015-1427 & T1210 service exploitation.\\
                  & CVE-2014-6271 & T1210 service exploitation.\\
                  & CVE-2016-10033 & T1210 service exploitation.\\
                  & SQL injection & T1210 service exploitation. \\
  \bottomrule\\
\end{tabular}
}
\caption{Examples of attacker actions in CSLE; actions are identified by identifiers in the common vulnerabilities and exposures (CVE) database \cite{cve}; the actions are also linked to the corresponding attack techniques in the MITRE ATT\&CK taxonomy \cite{strom2018mitre}.}\label{tab:dt_attacker_actions}
\end{table}

\subsection{The Management System}\label{sec:management}
The role of the management system in CSLE is to support the operation of digital twins and facilitate end-to-end reinforcement learning experiments. In particular, the management system provides APIs for real-time monitoring and control of digital twins, as well as a web interface for managing reinforcement learning experiments and deployments.

Each emulated device in a digital twin runs a \textit{management agent}, which exposes a GRPC API \cite{grpc}. This API is invoked to perform control actions, e.g., restarting services and updating configurations. The communication channels to the agents are provided by a \textit{management network}. The reason for using a separate network to carry management traffic is to avoid interference and simplify control of the digital twin \cite{clemm2007network}.

We provide an example of using the management system to execute control actions inside a digital twin in Listing \ref{lst:control}. To complement the Python APIs, the management system also includes a web interface and a command-line interface, both of which provide the same functions as the Python API. A video demonstration of the management system is available at \cite{csle_docs} and screenshots of the web interface are provided in Appendix~\ref{app:frontend} in the supplementary material.

\begin{lstlisting}[caption={Python code for executing a control action.}, label={lst:control}]
from csle_common.metastore.metastore_facade import MetastoreFacade  
from csle_common.util.emulation_util import EmulationUtil
twin=MetastoreFacade.get_twin(name="")
EmulationUtil.execute_ssh(cmds=[cmd], ip="", twin=twin)
\end{lstlisting}

To monitor processes and services running inside the digital twin, we use a monitoring system based on a publish-subscribe architecture; see Fig. \ref{fig:monitoring}. Following this architecture, each emulated device in a digital twin runs a \textit{monitoring agent}, which reads local metrics of the host and pushes those metrics to an event bus implemented with Kafka \cite{kafka}. The data in this bus is consumed by data pipelines, which process the data and write it to storage systems. In particular, the data is exported to an Elasticsearch database that can be queried and visualized through Kibana dashboards for real-time monitoring.

\begin{figure}[H]
  \centering
  \scalebox{0.8}{
    \input{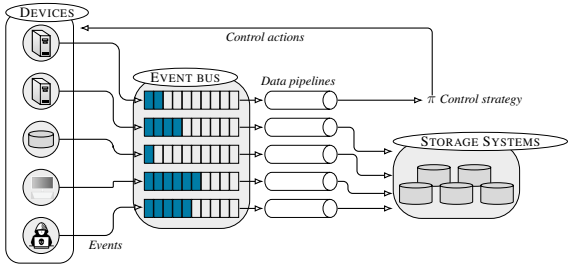}
    }
    \caption{Monitoring system of a digital twin in CSLE. Emulated devices run monitoring agents that periodically push metrics to an event bus, which is consumed by pipelines that process the data and write to storage systems; the processed data is also used as input to automated control strategies to decide on control actions.}
    \label{fig:monitoring}
  \end{figure}
  \begin{table}[H]
  \centering
\scalebox{0.7}{
  \begin{tabular}{lll} \toprule
\rowcolor{lightgray}    
  {\textit{Action}} & {\textsc{mitre d3fend} technique}\\ \midrule
  Revoke user certificates & D3-CBAN certificate revocation. \\
  Blacklist IPs & D3-NTF network traffic filtering. \\
  Drop traffic & D3-NTF network traffic filtering. \\
  Block gateway & D3-NI network isolation. \\
  Migrate servers between zones & D3-NI network isolation. \\
  Redirect traffic & D3-NTF network traffic filtering. \\
  Isolate a server & D3-NI network isolation. \\
  Deploy new security functions & D3-NTPM network policy mapping. \\
  Shutdown a server & D3-HS host shutdown. \\
  Replicate a service & D3-SVCDM service mapping. \\
  Start decoy services & D3-D3 decoy environment. \\
  \bottomrule\\
\end{tabular}
}
\caption{Examples of defender actions in CSLE; the actions are linked to the corresponding defense techniques in the MITRE D3FEND taxonomy \cite{kaloroumakis2021toward}.}\label{tab:defender_stop_actions}
\end{table}

Figure \ref{fig:monitoring_overhead} shows performance statistics related to the monitoring system. In particular, Fig. \ref{fig:monitoring_overhead}.a shows that the CPU overhead introduced by the monitoring agents is around 6\%, while the memory overhead is approximately 1\%. Both values can be considered relatively low. Furthermore, Fig. \ref{fig:monitoring_overhead}.b shows that the number of monitoring events produced by the monitoring agents per monitoring interval increases with the size of the digital twin and also depends on the specific system configuration. Larger twins typically contain more monitored components, which naturally results in a higher number of generated events. In addition, the event rate is influenced by the types and number of monitoring mechanisms deployed. For example, configurations that include a larger number of intrusion detection systems (as is the case for the system in Fig.~\ref{fig:systems}.b) generate more monitoring events.

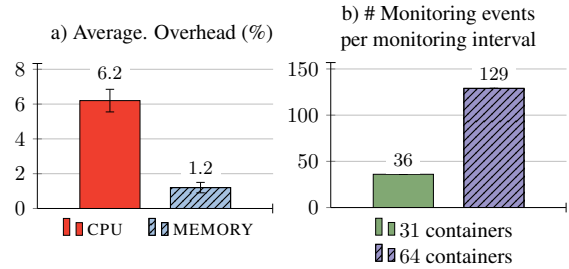
\begin{figure}
  \centering
  \scalebox{0.8}{
    \begin{tikzpicture}    

\node[scale=1] (kth_cr) at (0,0)
{
\begin{tikzpicture}
\begin{axis}[
   ybar,
    title style={align=center},
    ticks=both,
    ymin=0,
    axis x line = bottom,
    axis y line = left,
    axis line style={-|},
    enlarge y limits={lower, value=0.1},
    enlarge y limits={upper, value=0.22},
    xtick=\empty,
    ymajorgrids,
    xticklabels={},
    legend style={nodes={scale=1, transform shape}, at={(0.08, -0.15)}, align=left, anchor=west, legend columns=4, draw=none},
   x tick label style={align=center, yshift=-0.1cm},
    enlarge x limits=4,
    width=5.5cm,
    bar width=1cm,
    height=4cm,
    ]

\addplot+[
  draw=black, color=black,fill=Red!90,
  nodes near coords,
  every node near coord/.append style={
      anchor=south,
      shift={(axis direction cs:0,0.8)}, 
      font=\small\bfseries, fill=white,
scale=1,
  },
  error bars/.cd, y dir=both, y explicit,
] coordinates {
  (0,6.2) +- (0,0.65)
};

\addplot+[
  draw=black, color=black,fill=bluethree,
postaction={pattern=north east lines},
  nodes near coords,
  every node near coord/.append style={
      anchor=south,
      shift={(axis direction cs:0,0.38)}, 
      font=\small\bfseries, fill=white,
scale=1,
  },
  error bars/.cd, y dir=both, y explicit,
] coordinates {
  (1,1.2) +- (0,0.3)
};

\legend{\textsc{cpu}$\quad$, \textsc{memory}}
\end{axis}
\end{tikzpicture}
};

\node[scale=1] (kth_cr) at (4.7,-0.225)
{
\begin{tikzpicture}
\begin{axis}[
   ybar,
    title style={align=center},
    ticks=both,
    ymin=0,
    axis x line = bottom,
    axis y line = left,
    axis line style={-|},
    enlarge y limits={lower, value=0.1},
    enlarge y limits={upper, value=0.22},
    xtick=\empty,
    ymajorgrids,
    xticklabels={},
    legend style={nodes={scale=1, transform shape}, at={(0.16, -0.25)}, align=left, anchor=west, legend columns=1, draw=none},
   x tick label style={align=center, yshift=-0.1cm},
    enlarge x limits=4,
    width=5.5cm,
    bar width=1cm,
    height=4cm,
    ]

\addplot+[
  draw=black, color=black,fill=OliveGreen!60,
  nodes near coords,
  every node near coord/.append style={
      anchor=south,
      shift={(axis direction cs:0,0.8)}, 
      font=\small\bfseries, fill=white,
scale=1,
  },
  error bars/.cd, y dir=both, y explicit,
] coordinates {
  (0,36) +- (0,0.0)
};

\addplot+[
  draw=black, color=black,fill=Blue!40,
postaction={pattern=north east lines},
  nodes near coords,
  every node near coord/.append style={
      anchor=south,
      shift={(axis direction cs:0,1.35)}, 
      font=\small\bfseries, fill=white,
scale=1,
  },
  error bars/.cd, y dir=both, y explicit,
] coordinates {
  (1,129) +- (0,0.0)
};

\legend{$31$ containers, $64$ containers}
\end{axis}
\end{tikzpicture}
};

\node[inner sep=0pt,align=center, scale=1, rotate=0, opacity=1] (obs) at (0.4,1.95)
{
a) Average. Overhead (\%)
};

\node[inner sep=0pt,align=center, scale=1, rotate=0, opacity=1] (obs) at (4.95,2.05)
{
  b) \# Monitoring events\\ per monitoring interval
};


\end{tikzpicture}
 }
 \caption{Statistics of the monitoring system in CSLE. Plot a) shows the average overhead of a monitoring agent and plot b) shows the number of monitoring events per monitoring interval for two digital twins deployed with CSLE. The network topologies of the digital twins with $31$ and $64$ containers are shown in Fig.~\ref{fig:systems}.a and Fig.~\ref{fig:systems}.b, respectively. Numbers and error bars indicate the mean and the standard deviation from $5$ evaluations.}
\label{fig:monitoring_overhead}
\end{figure}

\begin{figure*}
  \centering
  \scalebox{0.81}{
   \input{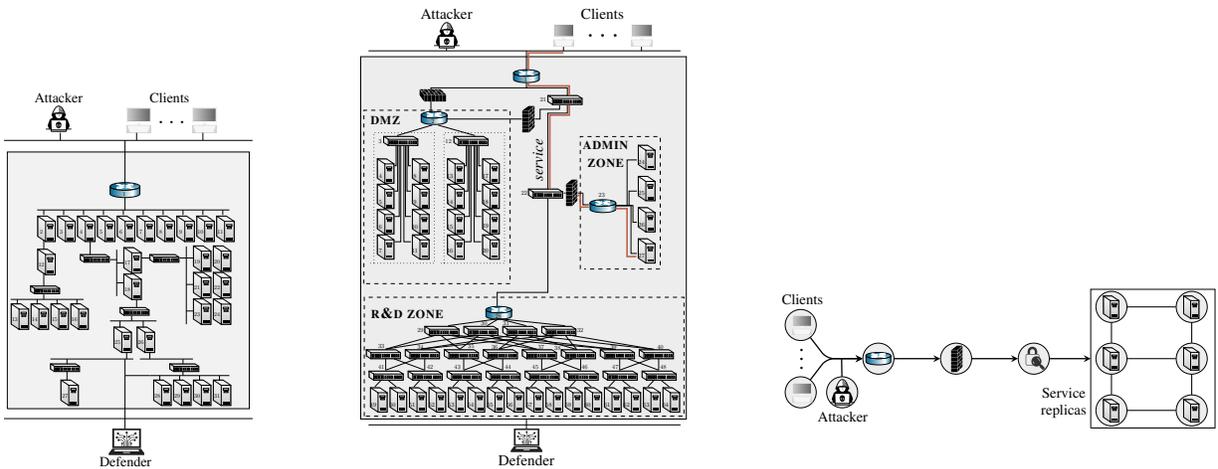}    
  }
  \caption{Target systems for the use cases in the experimental evaluation. The system configurations are available in Appendix~\ref{app:target_systems}.}
  \label{fig:systems}
\end{figure*}

\subsection{The Simulation System}\label{sec:simulation}
The simulation system in CSLE is implemented in Python and consists of reinforcement learning environments and algorithms for learning security strategies. All environments follow the OpenAI Gym interface \cite{towers2024gymnasiumstandardinterfacereinforcement}, which allows integration with standard reinforcement learning frameworks. Each environment defines a Markov decision process or a game and can be configured through Python configuration files. CSLE includes an initial set of 34 reinforcement learning algorithms, over 50 simulation environments, and 4 identification algorithms. We provide an example of using the simulation system to run a reinforcement learning algorithm in Listing \ref{lst:rl_sim}.

\begin{lstlisting}[caption={Python code for running the SARSA reinforcement learning algorithm in the simulation system.}, label={lst:rl_sim}]
from csle_agents.agents.sarsa.sarsa_agent import SARSAAgent
from csle_common.metastore.metastore_facade import MetastoreFacade
from csle_common.dao.training.experiment_config import ExperimentConfig
simulation = MetastoreFacade.get_simulation(..)
experiment = ExperimentConfig(..)
agent = SARSAAgent(simulation, experiment)
execution = agent.train()
MetastoreFacade.save_experiment_execution(execution)
for strategy in execution.result.strategies.values():
    MetastoreFacade.save_strategy(strategy)
\end{lstlisting}

\section{Example Use Cases}
We demonstrate CSLE by applying it to four different security use cases. Each use case involves a target system and a system operator, which we refer to as the \textit{defender}; see Fig. \ref{fig:systems}. (The detailed system configurations are available in Appendix~\ref{app:target_systems}.) The use cases are described below.

\subsection{Flow Control}
This use case involves an IT system that provides services to clients through a public gateway; see Fig.~\ref{fig:systems}.a. While the gateway enables legitimate access for clients, it also exposes an entry point for potential attackers attempting to intrude on the system and compromise components. To protect the system against such intrusions, the defender continuously monitors network traffic and analyzes security alerts to identify suspicious or malicious activity. Based on these observations, the defender can control network flows to mitigate potential network intrusions. For example, the defender can block suspicious flows or redirect them to a honeypot. When making these decisions, the defender aims to balance the dual objectives of preserving service availability for clients and mitigating potential attacks.

\subsection{Network Segmentation Control}
This use case involves a cloud infrastructure that is segmented into zones with virtual nodes that run network services; see Fig. \ref{fig:systems}.b. Services are realized by \textit{workflows} that clients access through a cloud gateway, which is also open to an attacker. The attacker aims to intrude on the infrastructure, compromise nodes, and disrupt workflows. To counter these threats, the defender continuously monitors the infrastructure by accessing and analyzing intrusion detection alerts and other statistics. Based on this information, the defender can respond to possible intrusions by changing the network segmentation. For example, the defender can migrate nodes between zones, change access controls, or shut down nodes. When deciding between these actions, the defender balances two conflicting objectives: maximizing workflow utility towards clients and minimizing the operational cost of possible intrusions and defensive actions.

\subsection{Recovery Control}
This use case involves a replicated system that provides a web service to a client population; see Fig. \ref{fig:systems}.c. Because multiple replicas deliver the same service, the system can continue operating even when some replicas are compromised. To track the evolving security status of the system, the defender analyzes security alerts that indicate potential compromises of service replicas. Based on these observations, the defender decides when and which replicas to recover in order to maintain service availability. When making these recovery decisions, the goal is to ensure that compromised service replicas are recovered faster than new compromises occur while minimizing the recovery costs.

\subsection{Replication Control}
In this use case, we study the problem of learning adaptive replication strategies for the system illustrated in Fig.~\ref{fig:systems}.c. The goal is to enable the system to autonomously adjust the number of replicas in response to changing security and performance conditions. In particular, the defender observes indicators of replica failures or degradations (e.g., security alerts) and uses these signals to dynamically decide when to launch new replicas or retire existing ones. By adapting the replication factor, the system can maintain service availability even under fluctuating attack intensities or workload demands. The key challenge is to minimize the number of replicas (to reduce operational costs) while still satisfying service availability and reliability constraints.
\begin{figure*}
  \centering
  \scalebox{0.76}{
   \input{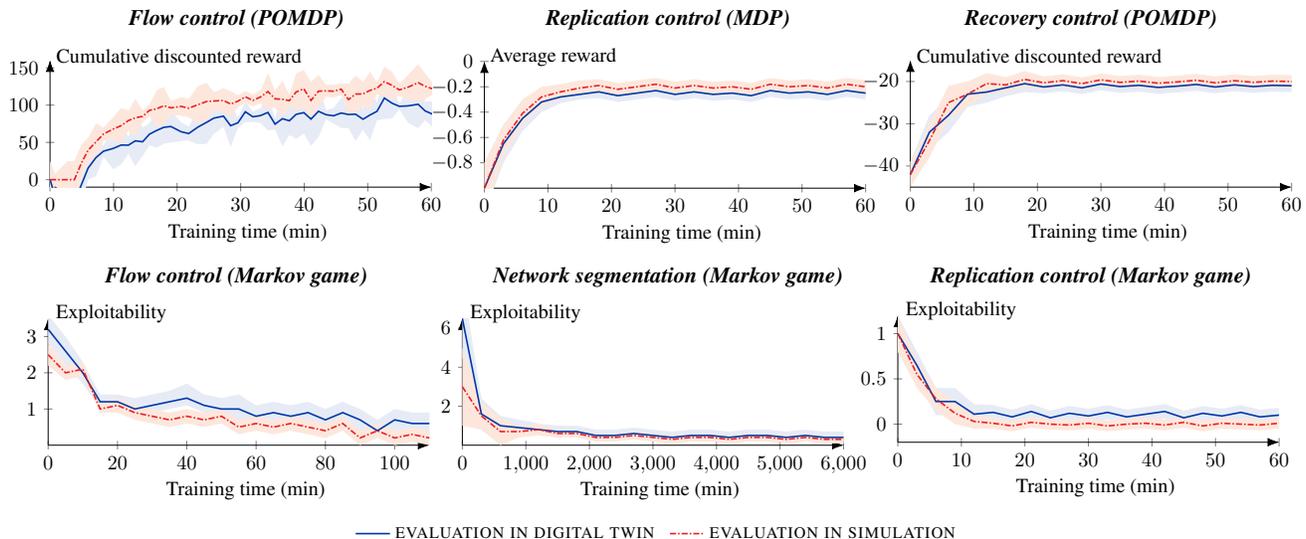}    
  }
  \caption{Convergence curves for the different use cases. The red curves relate to the performance in the simulations and the blue curves relate to the performance when evaluating the learned strategies in the digital twins. Curves show the mean values from evaluations with $5$ random seeds; shaded areas indicate standard deviations. The x-axes indicate the training times in the simulations. The top row relates to simulations of decision-theoretic models. The bottom row relates to simulations of game-theoretic models.}
  \label{fig:eval}
\end{figure*}

\section{System Models}
For each of the use cases described above, we consider two different reinforcement learning problems. First, we consider the problem of learning an optimal control strategy against an attacker that follows a fixed strategy. We model this problem as a Markov decision process (MDP) or a partially observed MDP (POMDP), depending on whether the system state is observable. Second, we consider the problem of learning a control strategy that is effective against a \textit{dynamic} attacker that adapts its strategy to circumvent the defenses. We model this problem as a Markov game. In total, we consider six reinforcement learning problems, which are listed in Table~\ref{tab:algos}. We address each problem using the general reinforcement learning methodology described in \S \ref{sec:end_to_end}. Mathematical formulations are given in Appendix \ref{app:models}.

\begin{table}[H]
  \centering
  \scalebox{0.76}{
    \begin{tabular}{ll} \toprule
\rowcolor{lightgray}      
  {\textit{Model}} & {\textit{Algorithm}} \\ \midrule
  Flow control POMDP & SPSA \cite{spsa} \\
  Replication control MDP & PPO \cite{ppo} \\
  Recovery control POMDP & Rollout \cite{bertsekas2021rollout} \\
  Flow control game & Fictitious play with SPSA \cite{brown_fictious_play} \\
  Segmentation game & Fictitious play with PPO \cite{brown_fictious_play} \\
  Replication control game & PPO \cite{ppo} \\  
  \bottomrule\\
\end{tabular}
}
\caption{The reinforcement learning problems that we consider in the experimental evaluation and the algorithms that we use to address them. In the decision-theoretic problems (i.e., the MDP and POMDPs), the goal is to learn an optimal security strategy against a fixed attacker strategy, whereas in the game-theoretic problems, the goal is to learn an equilibrium strategy against an attacker that dynamically adapts its strategy to the defender's strategy.}\label{tab:algos}
\end{table}
\section{Experimental Evaluation of CSLE}\label{sec:experimental_eval}
In this section, we present our experimental results. Following the methodology described in \S \ref{sec:end_to_end}, we identify the parameters of each system model described in the preceding section (e.g., the MDP, POMDP, or game parameters) based on data collected from the digital twin. Then, given the identified model, we learn the security strategy through simulation, after which we evaluate the learned strategy in the digital twin. The reinforcement learning algorithms that we use are listed in Table~\ref{tab:algos}. For details about the system identification and the collected data, see Appendix \ref{app:models}.

\subsection{Experimental Setup}
We run the reinforcement learning algorithms until convergence and evaluate the learned strategies periodically during training, both in the simulation and in the digital twin. The hyperparameters that we use to instantiate the algorithms are listed in Appendix~\ref{app:hyperparameters}. The environment for training strategies and running simulations is a Tesla P100 GPU. The digital twins are deployed on a server with a 24-core Intel Xeon Gold 2.10 GHz CPU and 768 GB RAM.

We consider two evaluation metrics: the reward and the exploitability. The reward measures the overall performance of the learned strategy in terms of its ability to achieve the defender’s objective in each use case. For the definitions of the reward functions, see Appendix \ref{app:models}. The exploitability, on the other hand, measures the distance of the learned strategies from a Nash equilibrium, where an exploitability of $0$ means that the strategies are in equilibrium.
\begin{figure*}
  \centering
  \scalebox{0.76}{
   \input{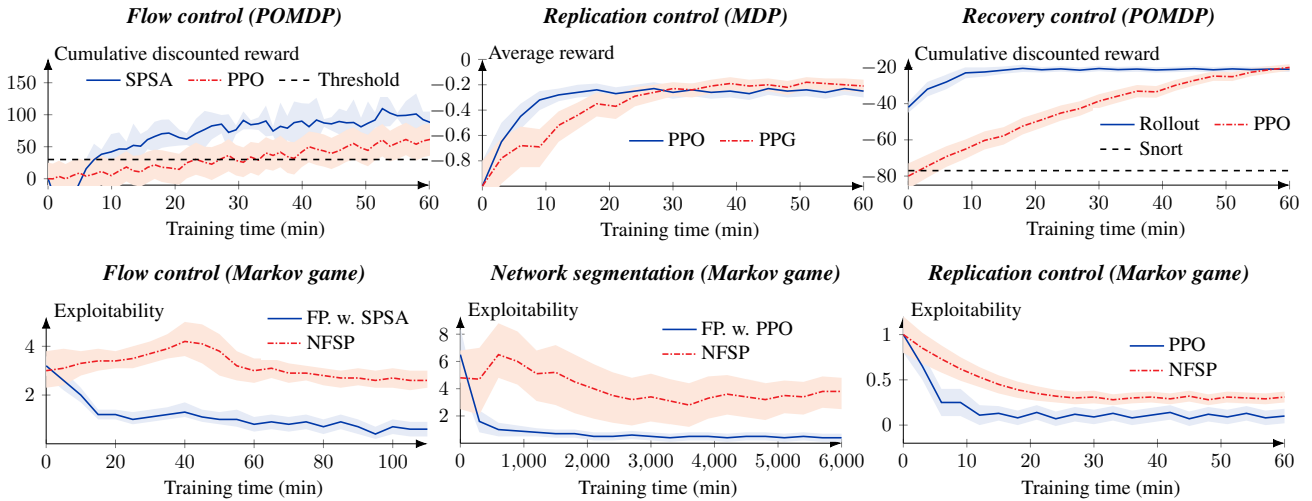}    
  }
  \caption{Performance comparison between reinforcement learning methods and baseline strategies in the digital twin. Curves show the mean values from evaluations with $5$ random seeds; shaded areas indicate standard deviations. The top row relates to the decision-theoretic models. The bottom row relates to the game-theoretic models. The acronym FP stands for fictitious play.}
  \label{fig:eval_2}
\end{figure*}

\subsection*{Baselines}
We compare the reinforcement learning algorithms listed in Table~\ref{tab:algos} against two baseline methods, namely phasic policy gradient (PPG) \cite{ppg} and neural fictitious self-play (NFSP) \cite{heinrich2016deepreinforcementlearningselfplay}. Additionally, we compare the performance of the reinforcement learning methods with that of two static security strategies: a threshold flow control strategy and a recovery strategy based on the Snort intrusion detection and prevention system with community ruleset v2.9.17.1 \cite{snort}.

The threshold strategy blocks network flows when the defender's belief (probability) that the system is compromised exceeds a predefined threshold $\alpha = 0.75$. The belief is computed according to the POMDP model, which maintains a probabilistic estimate of the underlying system state based on the sequence of observed system events. When the belief that the system is in a compromised state surpasses the threshold, the strategy proactively blocks network flows in order to limit potential attacker movement and prevent further propagation of the compromise.

The Snort baseline follows a rule-based recovery strategy that represents a signature-based response mechanism. This strategy recovers a component of the target system (e.g., by redeploying it in a new virtual machine) when a Snort alert with priority medium or higher is generated.
\subsection{Evaluation Results}
The evaluation results are summarized in Figs.~\ref{fig:eval}--\ref{fig:eval_2}. The red and blue curves in Fig.~\ref{fig:eval} represent the results from the simulator and the digital twin, respectively. An analysis of these curves leads us to the following conclusions. The learning curves converge to nearly constant mean values for all use cases and evaluation metrics. From this observation, we conclude that the learned strategies have also converged.

Although the learned strategies, as expected, perform slightly better on the simulator than on the digital twin, we are encouraged by the fact that the curves of the digital twin are close to those of the simulator (cf. the blue and red curves). This small performance gap reflects inevitable discrepancies between the simulation model and the digital twin, such as differences in network latency, background processes, or unmodeled system dynamics.

Figure~\ref{fig:eval_2} shows a comparison between different reinforcement learning methods and baseline security strategies. We observe that the performance varies substantially. Overall, the results show that the reinforcement learning strategies significantly outperform the static strategies.
\subsection{Discussion}
The experimental results demonstrate that CSLE effectively enables the transfer of reinforcement learning-based security strategies from simulation to a digital twin that closely approximates an operational environment. The small performance gap observed between the simulator and the digital twin indicates that the identified simulation models capture the main dynamics of the target system. This transferability is a key step toward operational deployment, as it validates that security strategies learned in simulation remain effective when tested under realistic operating conditions.

For safety and operational reasons, we have not evaluated the learned strategies in a production environment. While the digital twin replicates the software, configuration, and timing behavior of a production environment, further study is needed to determine whether the learned strategies maintain comparable performance in a production environment.
\subsection{Sensitivity to Model Misspecification}
The effective transfer of the learned strategies from simulation to the digital twin indicates that the identified simulation dynamics capture the main characteristics of the target system. However, in practice, model misspecification may still arise due to factors such as measurement noise or changes in system behavior over time (e.g., data drift). To better understand the impact of such modeling inaccuracies, in this section, we analyze the sensitivity of the learned strategies to the misspecification of the system model.

From a theoretical perspective, the performance loss of learned strategies due to deviations between the model and the system dynamics is upper-bounded as follows.
\begin{proposition}[Model misspecification error bound]\label{prop:misspecification_bound}
Let $\tilde{f}$ denote the dynamics of the system model [cf.~\eqref{eq:dt_system}] and let $f$ denote the dynamics of the target system. Denote by $J_{\bm{\pi}}$ and $\tilde{J}_{\bm{\pi}}$ the value functions (i.e., the expected rewards) under a strategy pair $\bm{\pi}= (\pi_{\mathrm{D}},\pi_{\mathrm{A}})$ in the target system and in the simulation, respectively. If the dynamics satisfy
\begin{align*}
\sum_{s^{\prime} \in \mathcal{S}} \left|f(s^{\prime} \mid s, a^{(\mathrm{D})}, a^{(\mathrm{A})}) - \tilde{f}(s^{\prime} \mid s, a^{(\mathrm{D})}, a^{(\mathrm{A})}) \right| \leq \alpha,
\end{align*}
for all states and actions, and some constant $\alpha$. Then
\begin{align*}
\norm{J_{\bm{\pi}} - \tilde{J}_{\bm{\pi}}}_{\infty} \leq \frac{\alpha \gamma \beta}{(1-\gamma)^2},
\end{align*}
where
$\gamma <1 $ is the discount factor and $\beta$ is defined by
\begin{align*}
\beta = \max_{s \in \mathcal{S}, a^{(\mathrm{D})} \in \mathcal{A}_{\mathrm{D}},a^{(\mathrm{A})} \in \mathcal{A}_{\mathrm{A}} }|r(s,a^{(\mathrm{D})}, a^{(\mathrm{A})})|,
\end{align*}
where $r$ is the reward function, $\mathcal{S}$ is the state space, and $(\mathcal{A}_{\mathrm{D}},\mathcal{A}_{\mathrm{A}})$ are the action spaces.
\end{proposition}
We present the proof of Prop.~\ref{prop:misspecification_bound} in the supplementary material (Appendix~\ref{app:misspecification_bound}). This proposition states that the misspecification error grows proportionally with the error of the discrepancy between the state transitions in the system model and the digital twin, as quantified by the parameter $\alpha$. In practice, this parameter can be estimated by comparing the simulated state trajectories of the model and the trajectories observed in the digital twin.

While Prop.~\ref{prop:misspecification_bound} provides a worst-case bound on the impact of model misspecification, it is conservative and does not necessarily reflect the sensitivity of a specific system instance. To complement the theoretical analysis, we therefore conduct an empirical sensitivity analysis on the flow control model used in our experiments. In this model, the probability that the attacker successfully compromises the system is governed by a parameter $p \in [0,1]$; see Appendix~\ref{sec:flow_pomdp} in the supplementary material for details. In the digital twin, we configure this parameter as $p=0.01$. To introduce model misspecification, we vary $p$ in the simulation model and evaluate how the performance of the learned defender strategy changes as the discrepancy between the simulation model and the digital twin increases. The results of this sensitivity analysis are summarized in Fig.~\ref{fig:sensitivity}.
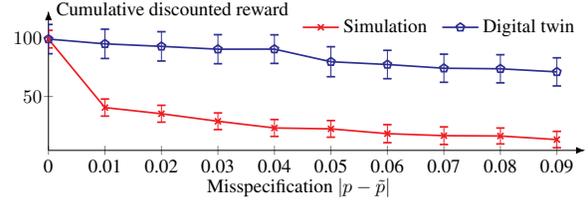
\begin{figure}[H]
  \centering
  \scalebox{0.8}{
    \begin{tikzpicture}
  \pgfplotstableread{
0.0 98.99 98.99 0.0 12.5 7.5
0.01 40.55 94.95 61.39 12.5 7.3
0.02 35.27 92.79 64.52 12.5 7.2
0.03 28.85 90.37 76.52 12.4 7.1
0.04 23.13 90.44 85.31 12.2 7.2
0.05 22.35 79.68 49.33 12.8 7.1
0.06 18.27 77.30 59.03 12.1 7.7
0.07 16.55 74.13 63.58 12.0 7.5
0.08 16.29 73.66 61.37 12.0 6.8
0.09 13.23 71.02 57.79 12.0 7.0
}\resourceone

\pgfplotsset{/dummy/workaround/.style={/pgfplots/axis on top}}

\node[scale=1] (kth_cr) at (0,0)
{
\begin{tikzpicture}
  \begin{axis}
[
        xmin=0,
        xmax=0.095,
        ymax=123,
        width=10.5cm,
        height=3.9cm,
        scaled x ticks=false,
        ymin=4,
        axis y line=center,
        axis x line=bottom,
        scaled y ticks=false,
        xticklabels={0, 0, 0.01, 0.02, 0.03, 0.04, 0.05, 0.06, 0.07, 0.08, 0.09},
        yticklabel style={
        /pgf/number format/fixed,
        /pgf/number format/precision=5
      },
        xlabel style={below right},
        ylabel style={above left},
        x tick label style={align=center, yshift=0.04cm, scale=0.9},
        y tick label style={align=center, xshift=0.08cm, scale=0.9},                
        axis line style={-{Latex[length=1.5mm]}},
        legend style={at={(1,1)}},
        legend columns=2,
        legend style={
          draw=none,
          nodes={scale=0.9, transform shape},
      anchor=north east,
      /tikz/every even column/.style={anchor=west, column sep=5pt},
      /tikz/every odd column/.style={anchor=west},          
            /tikz/column 2/.style={
                column sep=5pt,
              }
            },
            error bars/y dir=both,
error bars/y explicit,
error bars/error mark options={
  rotate=90, 
  mark size=2pt,
  thick
},
              ]
            \addplot[Red,name path=l1, thick, mark=x, mark repeat=1] table [x index=0, y index=1, y error index=5] {\resourceone};
            \addplot[Blue,name path=l1, thick, mark=pentagon, mark repeat=1] table [x index=0, y index=2, y error index=4] {\resourceone};

            \legend{Simulation, Digital twin}
            \end{axis}
\node[inner sep=0pt,align=center, scale=0.9, opacity=1] (obs) at (2.1,2.35)
{
Cumulative discounted reward
};
\node[inner sep=0pt,align=center, scale=0.9, rotate=0, opacity=1] (obs) at (4.2,-0.62)
{
  Misspecification $|p-\tilde{p}|$
};
\end{tikzpicture}
};

\end{tikzpicture}
 }
 \caption{Analysis of the sensitivity to model misspecification in the flow control use case. Numbers and error bars indicate the mean and the standard deviation from $5$ evaluations.}
\label{fig:sensitivity}
\end{figure}

Figure~\ref{fig:sensitivity} shows that small modeling errors lead to noticeable differences between the performance predicted by the simulation and the performance observed in the digital twin. However, these discrepancies in the simulated performance translate only to small variations in the performance of the learned defender strategy when it is evaluated in the digital twin. This indicates that, although the model is sensitive to misspecification when estimating performance, the learned strategy itself is relatively robust to such errors.
\section{Conclusion}
In this paper, we present CSLE, a comprehensive research platform for autonomous security management through reinforcement learning. This platform addresses the system-level challenges that arise in the operation and experimentation with reinforcement learning in networked systems. In particular, it is based on a novel methodology for learning security strategies that combines a digital twin with system identification and simulation-based reinforcement learning. Our evaluation across four security use cases demonstrates that this methodology narrows the gap between simulated and practical performance of learned security strategies.

\subsection*{Future Work}
Future work will focus on further developing the CSLE platform and expanding its open-source ecosystem. We plan to continue improving the platform’s usability by adding more learning resources, documentation, and example configurations to facilitate adoption by both researchers and practitioners. So far, our experimental validation of CSLE has focused primarily on IT systems. In future work, we aim to extend CSLE to cyberphysical systems.

\bibliography{references}
\bibliographystyle{mlsys2025}

\clearpage

\appendix

\section*{Supplementary Material}
This document provides supplementary material for the paper ``\textit{CSLE: A Reinforcement Learning Platform for Autonomous Security Management}'', published at the Machine Learning and Systems (MLSys) conference 2026. It contains the following appendices.

\begin{itemize}
\item \textbf{Appendix~\ref{app:frontend}}: \textit{Management System Frontend}.
\begin{itemize}
\item This appendix provides example screenshots of the management system in CSLE.
\end{itemize}
\item \textbf{Appendix~\ref{app:target_systems}}: \textit{Target System Configurations}.
\begin{itemize}
\item This appendix provides the detailed configurations of the target systems in Fig.~\ref{fig:systems}.
\end{itemize}
\item \textbf{Appendix~\ref{app:hyperparameters}}: \textit{Hyperparameters}.
\begin{itemize}
\item This appendix provides the hyperparameters that we use to instantiate the reinforcement learning algorithms in the experimental evaluation.
\end{itemize}
\item \textbf{Appendix~\ref{app:models}}: \textit{Simulation Models}.
\begin{itemize}
\item This appendix provides details about the simulation models used in the experimental evaluation.
\end{itemize}
\item \textbf{Appendix~\ref{app:misspecification_bound}}: \textit{Proof of Prop.~\ref{prop:misspecification_bound}}.
\begin{itemize}
\item This appendix provides the proof of Prop.~\ref{prop:misspecification_bound}.
\end{itemize}
\item \textbf{Appendix~\ref{app:observation_example}}: \textit{Example: Observation Mapping}.
\begin{itemize}
\item This appendix provides an illustrative example of how we map log events (e.g., security alerts) to observations that are input to the reinforcement learning strategies.
\end{itemize}
\item \textbf{Appendix~\ref{app:digital_twin_config}}: \textit{Example: Digital twin configuration}.
\begin{itemize}
\item This appendix provides an example configuration of a digital twin in CSLE.
\end{itemize}
\item \textbf{Appendix~\ref{app:artifact_appendix}}: \textit{Artifacts appendix}.
\begin{itemize}
\item This appendix provides details about the artifacts associated with this paper.
\end{itemize}
\end{itemize}

\begin{table*}
\centering
\scalebox{0.8}{
  \begin{tabular}{llll}
\rowcolor{lightgray}    
  {\textit{ID (s)}} & {\textit{Operating system}} & {\textit{Services}} & {\textit{Vulnerabilities}} \\ \midrule
  $1$ & Ubuntu 20 & Snort (community ruleset v2.9.17.1), SSH & - \\
  $2$ & Ubuntu 20 & SSH, HTTP Erl-Pengine, DNS & Weak password\\
  $4$ & Ubuntu 20 & HTTP, Telnet, SSH & Weak password \\
  $10$ & Ubuntu 20 & FTP, MongoDB, SMTP, Tomcat, Teamspeak 3, SSH & Weak password \\
  $12$ & Debian Jessie & Teamspeak 3, Tomcat, SSH & CVE-2010-0426, Weak password \\
  $17$ & Debian Wheezy & Apache 2, SNMP, SSH & CVE-2014-6271 \\
  $18$ & Debian 9.2 & IRC, Apache2, SSH & SQL injection \\
  $22$ & Debian Jessie & ProFTPd, SSH, Apache2, SNMP & CVE-2015-3306 \\
  $23$ & Debian Jessie & Apache2, SMTP, SSH & CVE-2016-10033 \\
  $24$ & Debian Jessie & SSH & CVE-2015-5602, Weak password \\
  $25$ & Debian Jessie & Elasticsearch, Apache 2, SSH, SNMP & CVE-2015-1427\\
  $27$ & Debian Jessie & Samba, NTP, SSH & CVE-2017-7494\\
  $3$,$11$,$5$-$9$& Ubuntu 20 & SSH, SNMP, PostgreSQL, NTP & -\\
  $13-16$,$19-21$,$26$,$28-31$& Ubuntu 20 & NTP, IRC, SNMP, SSH, PostgreSQL & -\\
\end{tabular}
}
\caption{Configuration of the target system in Fig.~\ref{fig:systems}.a.}\label{tab:config1}
\end{table*}

\begin{table*}
\centering
\scalebox{0.8}{
  \begin{tabular}{l|l|l|l|l|l}
\rowcolor{lightgray}    
  {\textit{ID(s)}} & {\textit{Type}} & {\textit{Operating system}} & {\textit{Zone}} & {\textit{Services}} & {\textit{Vulnerabilities}} \\ \midrule
  $1$ & Gateway & Ubuntu 20 & - & Snort (ruleset v2.9.17.1), SSH, OpenFlow v1.3, Ryu SDN controller & -\\
  $2$ & Gateway & Ubuntu 20 & DMZ & Snort (ruleset v2.9.17.1), SSH, OVS v2.16, OpenFlow v1.3 & -\\
  $28$ & Gateway & Ubuntu 20 & R\&D & Snort (ruleset v2.9.17.1), SSH, OVS v2.16, OpenFlow v1.3 & -\\
  $3$,$12$ & Switch & Ubuntu 22 & DMZ & SSH, OpenFlow v1.3 , OVS v2.16& -\\
  $21,22$ & Switch & Ubuntu 22 & - & SSH, OpenFlow v1.3, OVS v2.16 & -\\
  $23$ & Switch & Ubuntu 22 & Admin & SSH, OpenFlow v1.3, OVS v2.16 & -\\
  $29$-$48$ & Switch & Ubuntu 22 & R\&D & SSH, OpenFlow v1.3, OVS v2.16 & -\\
  $13$-$16$ & Honeypot & Ubuntu 20 & DMZ & SSH, SNMP, PostgreSQL, NTP & -\\
  $17$-$20$ & Honeypot & Ubuntu 20 & DMZ & SSH, IRC, SNMP, SSH, PostgreSQL & -\\
  $4$ & App node & Ubuntu 20 & DMZ & HTTP, DNS, SSH & CWE-1391\\
  $5$, $6$ & App node & Ubuntu 20 & DMZ & SSH, SNMP, PostgreSQL, NTP & -\\
  $7$ & App node & Ubuntu 20 & DMZ & HTTP, Telnet, SSH & CWE-1391\\
  $8$ & App node & Debian Jessie & DMZ & FTP, SSH, Apache 2, SNMP & CVE-2015-3306\\
  $9$,$10$ & App node & Ubuntu 20 & DMZ & NTP, IRC, SNMP, SSH, PostgreSQL & -\\
  $11$ & App node & Debian Jessie & DMZ & Apache 2, SMTP, SSH & CVE-2016-10033\\
  $24$ & Admin system & Ubuntu 20 & Admin & HTTP, DNS, SSH & CWE-1391\\
  $25$ & Admin system & Ubuntu 20 & Admin & FTP, MongoDB, SMTP, Tomcat, Teamspeak 3, SSH & -\\
  $26$ & Admin system & Ubuntu 20 & Admin & SSH, SNMP, Postgres, NTP & -\\
  $27$ & Admin system & Ubuntu 20 & Admin & FTP, MongoDB, SMTP, Tomcat, Teamspeak 3, SSH & CWE-1391\\
  $49$-$59$ & Compute node & Ubuntu 20 & R\&D & Spark, HDFS & -\\
  $60$ & Compute node & Debian Wheezy & R\&D & Spark, HDFS, Apache 2, SNMP, SSH & CVE-2014-6271\\
  $61$ & Compute node & Debian 9.2 & R\&D & IRC, Apache 2, SSH & CWE-89\\
  $62$ & Compute node & Debian Jessie & R\&D & Spark, HDFS, Teamspeak 3, Tomcat, SSH & CVE-2010-0426\\
  $63$ & Compute node & Debian Jessie & R\&D & SSH, Spark, HDFS & CVE-2015-5602\\
  $64$ & Compute node & Debian Jessie & R\&D & Samba, NTP, SSH, Spark, HDFS & CVE-2017-7494\\
\end{tabular}
}
\caption{Configuration of the target system in Fig.~\ref{fig:systems}.b.}\label{tab:config2}
\end{table*}

\begin{table*}
  \centering
\scalebox{0.85}{
  \begin{tabular}{llll}
\rowcolor{lightgray}    
  {\textit{ID(s)}} & {\textit{Operating system}} & {\textit{Background services}} & {\textit{Vulnerabilities}}  \\ \midrule
  $1$  & Debian 9.2 & Apache2 & CWE-89\\
  $2$  & Debian Jessie & FTP & CVE-2015-3306\\
  $3$  & Ubuntu 20 & SSH, Spark & CWE-1391 \\
  $4$  & Debian Jessie & Phpmailer & CVE-2016-10033\\
  $5$  & Debian Wheezy & Nginx & CVE-2014-6271\\
  $6$  & Debian Jessie & SSH, GRPC & CWE-1391, CVE-2010-0426\\
  $7$  & Debian Jessie & SSH, Spring boot & CVE-2015-5602, CWE-1391\\
  $8$  & Debian Jessie & PostgreSQL, Samba & CVE-2017-7494\\
\end{tabular}
}
\caption{Configuration of the target system in Fig.~\ref{fig:systems}.c. (The number of replicas is dynamically scalable. When starting a new replica, its configuration is selected randomly from the table.)}\label{tab:config3}
\end{table*}

\section{Management System Frontend}\label{app:frontend}
Three screenshots of the web interface to the management system in CSLE are shown in Figs.~\ref{fig:screen1}--\ref{fig:debugger}. A video demonstration is available at \cite{csle_docs}. The backend is implemented as a REST API in Python based on the Flask framework. The frontend is implemented in JavaScript based on the React framework.
\begin{figure*}
  \centering
  \scalebox{0.45}{
    \includegraphics{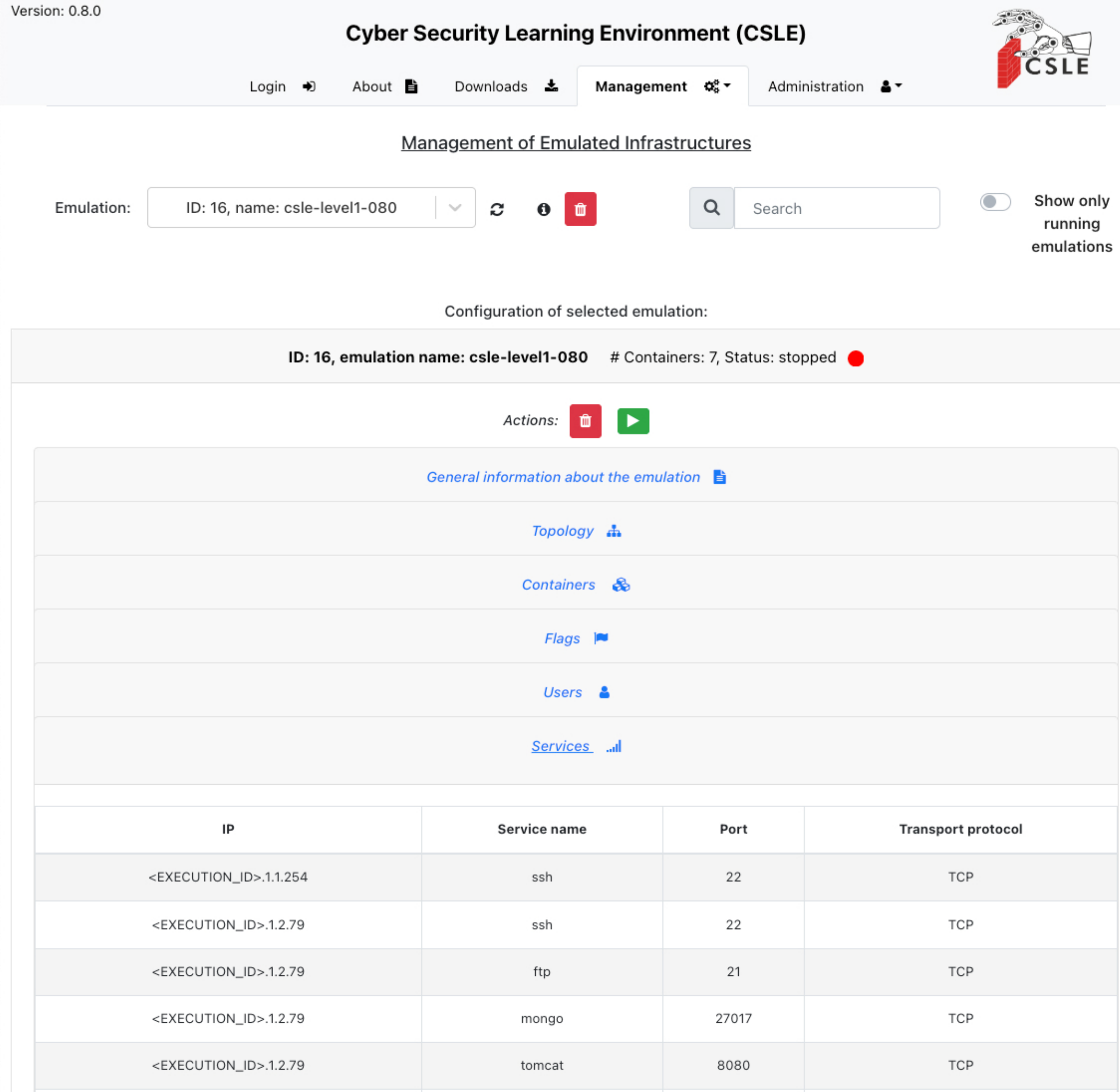}
  }
\caption{Screenshot of the web interface to the management system in CSLE. The figure shows the page for configuring emulations (i.e., digital twins). In addition to this page, the web interface allows viewing reinforcement learning experiments, debugging learned security strategies, real-time monitoring of digital twins, management of simulations, and integration with large language models.}
  \label{fig:screen1}
\end{figure*}

\begin{figure*}
  \centering
  \scalebox{0.21}{
    \includegraphics{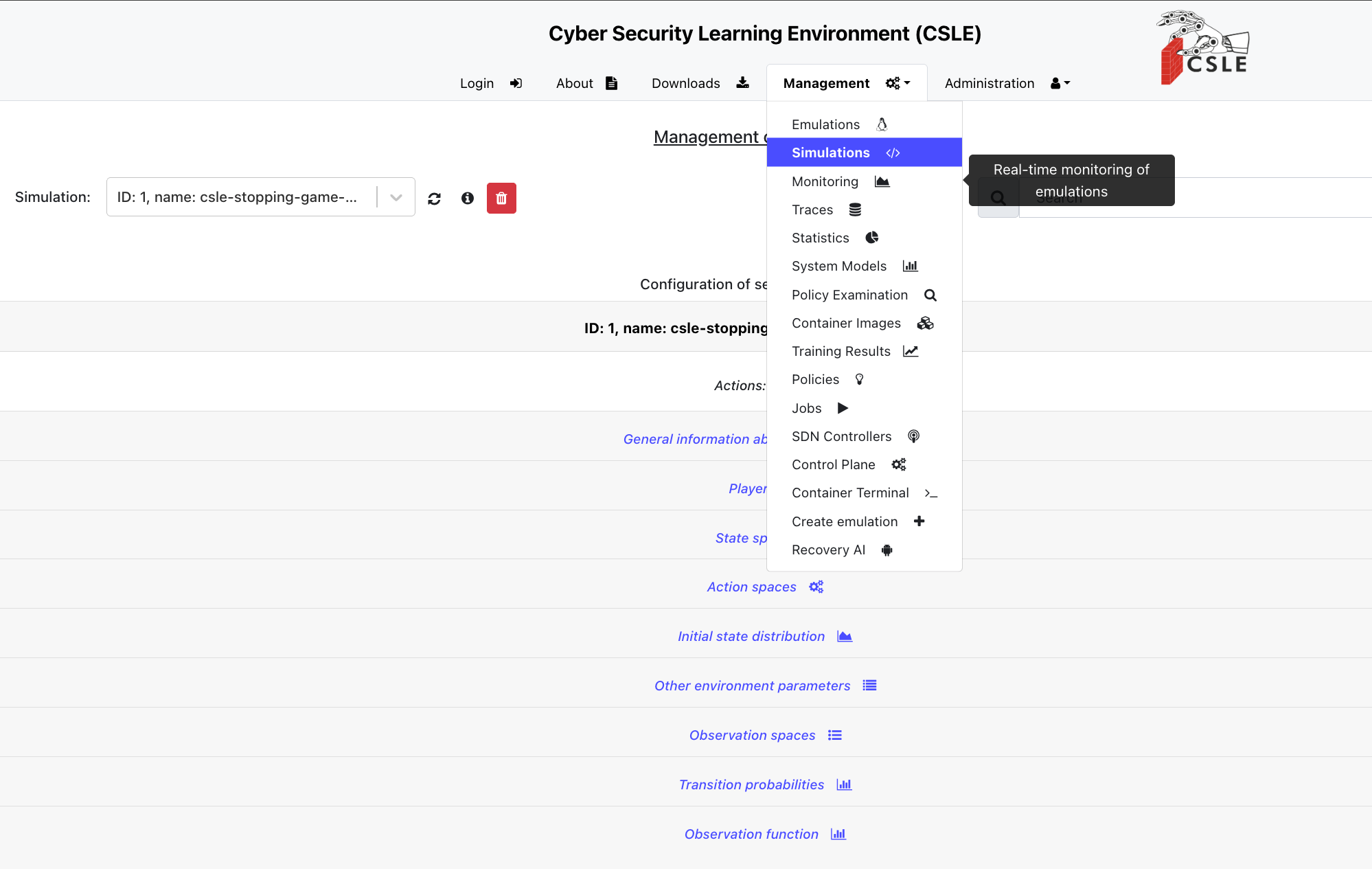}
  }
\caption{Screenshot of the web interface to the management system in CSLE. The figure shows the list of available pages.}
  \label{fig:screen3}
\end{figure*}

\begin{figure*}
  \centering
  \scalebox{0.38}{
    \includegraphics{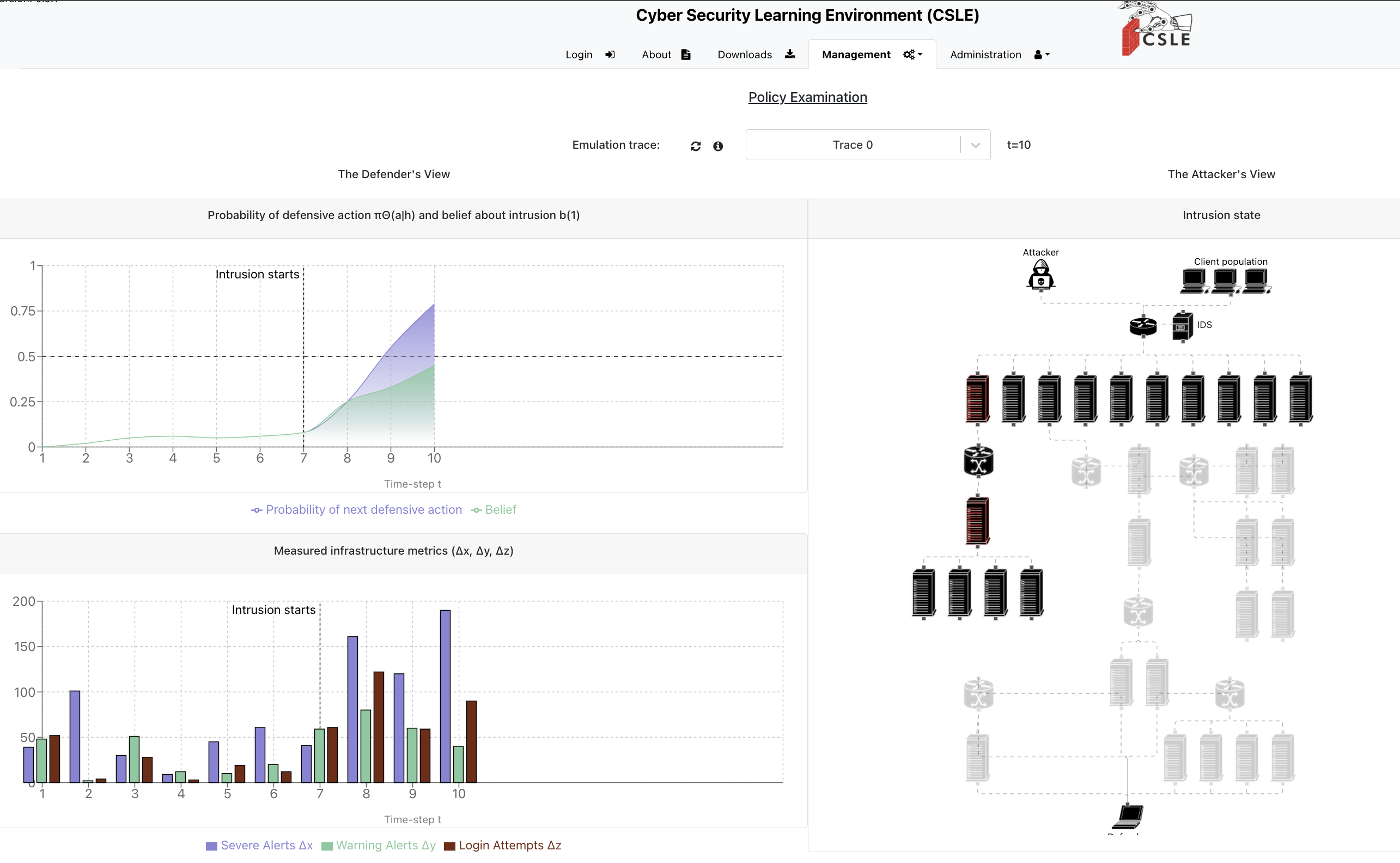}
  }
  \caption{Screenshot of the web interface of the strategy debugger in CSLE, which allows the user to interactively step through a POMDP episode. The left panel shows the defender's view, with infrastructure statistics updated in real time. The right panel shows the attacker's view, which consists of partial knowledge of the system under attack.}
  \label{fig:debugger}
\end{figure*}

\section{Target System Configurations}\label{app:target_systems}
The configuration of the target system in Fig.~\ref{fig:systems}.a is available in Table~\ref{tab:config1}. (Note that each component in Fig.~\ref{fig:systems} is labeled with an identifier $N_i$.) Similarly, the configuration of target systems in Fig.~\ref{fig:systems}.b and Fig.~\ref{fig:systems}.c are available in Table~\ref{tab:config2} and Table~\ref{tab:config3}, respectively.

\section{Hyperparameters}\label{app:hyperparameters}
The hyperparameters that we use to instantiate the reinforcement learning algorithms are listed in Table~\ref{tab:hyperparams}.

\begin{table}[H]
  \centering
  \scalebox{0.8}{
    \begin{tabular}{ll} \toprule
\rowcolor{lightgray}
      {\textit{Parameter(s)}} & {\textit{Value(s)}} \\ \midrule
      Discount factor $\gamma$ & $0.99$ \\\midrule
      
      \textbf{SPSA} \cite{spsa} &   \\\midrule
      $c,\epsilon,\lambda,A,a$ & $1,0.101, 0.602, 100, 1$\\
      \midrule\\
      \textbf{Rollout} \cite{bertsekas2021rollout} &   \\
      \midrule
      Rollout horizon & $20$\\
      Lookahead horizon & $1$\\
      Monte-Carlo samples & $20$\\       
    \midrule\\
    \textbf{PPO} \cite{ppo} &   \\
    \midrule      
    Learning rate, \# hidden layers  & $5.148 \cdot 10^{-5}$, $1$ \\
    \# Neurons/layer & $64$\\
    \# Steps between updates & $2048$\\      
    Batch size, discount factor $\gamma$ & $16$, $0.99$\\
    \textsc{gae} $\lambda$, clip range, entropy coefficient & $0.95$, $0.2$, $2\cdot 10^{-4}$\\
    Value coefficient, max gradient norm & $0.102$, $0.5$\\
    \bottomrule\\
  \end{tabular}}
  \caption{Hyperparameters used for the experimental evaluation.}\label{tab:hyperparams}
\end{table}
\section{Simulation Models}\label{app:models}
In this appendix, we detail the models used to instantiate the simulations for the use cases described in \S \ref{sec:experimental_eval}.
\subsection{Flow Control POMDP}\label{sec:flow_pomdp}
We formulate the flow control use case when the attacker follows a static strategy as a POMDP, which is defined by the following nine-tuple
\begin{align}
\langle \mathcal{S}, \mathcal{A}, f, r, \gamma, b_1, T, \mathcal{O}, z \rangle.\label{eq:background_pomdp_def}
\end{align}
The POMDP evolves in time steps from $t=1$ to $t=T$, which constitutes one \textit{episode}. The constant $\gamma \in [0,1]$ is a discount factor, $\mathcal{S}$ is the set of states, and $\mathcal{A}$ is the set of actions. The initial state is drawn from $b_1 \in \Delta(\mathcal{S})$ and $f(s_{t+1}\mid s_t, a_t)$ is the probability of transitioning from state $s_t$ to state $s_{t+1}$ when taking action $a_t$. The set of observations is denoted by $\mathcal{O}$ and $z(o_{t} \mid s_{t})$ is the \textit{observation function}, where $o_{t} \in \mathcal{O}$.

\begin{figure*}
  \centering`
    \scalebox{0.85}{
      \includegraphics{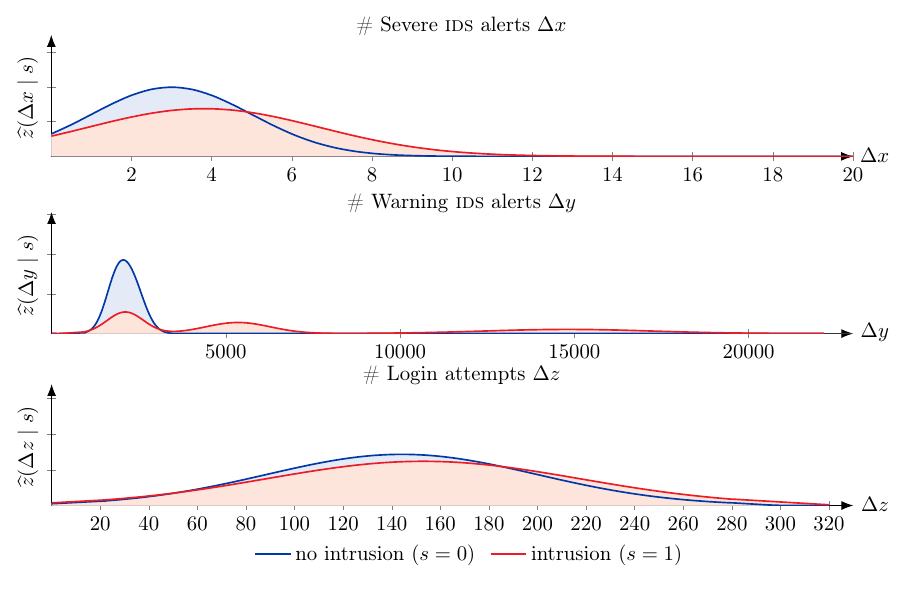}
    }
    \caption{Estimated (smoothed) distributions of severe IDS alerts $\Delta x$ (top row), warning IDS alerts $\Delta y$ (middle row), and login attempts $\Delta z$ (bottom row) for the flow control POMDP. The distributions are estimated based on measurements from the digital twin.}
    \label{fig:obs_dist_1}
  \end{figure*}
\paragraph{Actions \texorpdfstring{$\mathcal{A}$}{Lg}}\label{sec:p1_actions}
The defender has two actions: ($\mathsf{S}$)top and ($\mathsf{C}$)ontinue. The action space is thus $\mathcal{A} = \{\mathsf{S},\mathsf{C}\}$. We encode $\mathsf{S}$ with $1$ and $\mathsf{C}$ with $0$ to simplify the formal description below. Each stop is associated with a flow control action, and the objective is to decide the optimal times for stopping. The number of stops the defender must execute to prevent an intrusion is $L \geq 1$, which is a predefined parameter of our use case. The number of stop actions remaining is denoted by $l \in \{1,\hdots,L\}$.

\paragraph{States \texorpdfstring{$\mathcal{S}$}{Lg}}
The state $s_t$ is $0$ if no intrusion occurs and $1$ if an intrusion is ongoing. The terminal state $\emptyset$ is reached after the defender takes the final stop action. The state space is thus $\mathcal{S} = \{0,1,\emptyset\}$. The initial state is $s_1 = 0$. Hence, $b_1 \in \Delta(\mathcal{S})$ is the degenerate distribution $b_1(0) = 1$.

\paragraph{Observations \texorpdfstring{$\mathcal{O}$}{Lg}}
The defender has a partial view of the system and observes $o_t = (\Delta x_t, \Delta y_t, \Delta z_t)$, where $\Delta x_t$, $\Delta y_t$, and $\Delta z_t$ are bounded counters that denote the number of severe IDS alerts, warning IDS alerts, and login attempts generated during time step $t$, respectively. 

\paragraph{Transition function \texorpdfstring{$f_l(s^{\prime} \mid s, a)$}{Lg}}
We model the start of an intrusion by a Bernoulli process $(Q_t)_{t=1}^{T}$, where $Q_t \sim \mathrm{Ber}(p)$ is a Bernoulli random variable with $p > 0$. The first occurrence of $Q_t=1$ defines the intrusion start time $I$, which thus is geometrically distributed.

Consequently, we can define the transition function as
\begin{subequations}\label{eq:p1_transition_dynamics}
\begin{align}
&f_{1}(\emptyset \mid \cdot , 1) = f_{l}(\emptyset \mid \emptyset,\cdot)= 1, \label{eq:p1_tp_1}\\
&f_{l}(0 \mid 0, a) = 1-p, && \text{ if $l-a > 0$},\label{eq:p1_tp_2}\\
&f_{l}(1 \mid 0, a)= p, && \text{ if $l-a > 0$},\label{eq:p1_tp_3}\\
&f_{l}(1 \mid 1, a) = 1, && \text{ if $l-a > 0$},\label{eq:p1_tp_4}
\end{align}
\end{subequations}
where $a \in \mathcal{A}$.\footnote{Recall that we encode $(\mathsf{S}, \mathsf{C})= (1,0)$; hence $l_{t+1}=l_t-a_{t}$.}\footnote{We define $p=0.01$ for the evaluation reported in the paper.} All other state transitions occur with probability $0$. Equation (\ref{eq:p1_tp_1}) defines the transitions to the terminal state $\emptyset$, which is reached when the \textit{final} stop action is taken (i.e., when $l=1$ and $a = 1$). If (\ref{eq:p1_tp_1}) is not applicable, i.e., if the system does not reach the terminal state, then the transitions are defined by (\ref{eq:p1_tp_2})-(\ref{eq:p1_tp_4}). Equation (\ref{eq:p1_tp_2}) captures the case where no intrusion occurs; (\ref{eq:p1_tp_3}) specifies the case when the intrusion starts; and (\ref{eq:p1_tp_4}) describes the case where an intrusion is in progress. Note that the intrusion state $s=1$ is absorbing until $L$ stop actions have been taken.

\paragraph{Observation function \texorpdfstring{$z(o_{t} \mid s_{t})$}{Lg}}\label{sec:p1_obs_fun}
We estimate the distributions of IDS alerts and login attempts using data from the digital twin. At the end of every $30$s interval on the twin, we collect the metrics $\Delta x$, $\Delta y$, $\Delta z$, which contain the alerts and login attempts that occurred during the interval. We use $M=21,000$ i.i.d. samples to compute the empirical distribution $\widehat{z}(\cdot \mid s_t)$ as an estimate of $z(\cdot \mid s_t)$, where $\widehat{z} \overset{\text{a.s.}}{\rightarrow} z$ as $M \rightarrow \infty$ \cite{glivenko_cantelli}. Figure~\ref{fig:obs_dist_1} shows some of the estimated (smoothed) distributions. The distributions during normal operation and intrusion overlap. However, the distributions during intrusions tend to have more probability mass at larger values of $\Delta x, \Delta y$, and $\Delta z$.

\paragraph{Reward function \texorpdfstring{$r(s, a)$}{Lg}}\label{sec:p1_reward_fun}
The objective is to maintain service on the infrastructure while preventing a possible intrusion. Therefore, we define the reward function to give the maximum reward if the defender maintains service until the intrusion starts and then prevents the intrusion by taking $L$ stop actions. The reward per time step $r(s,a)$ is parameterized by the reward that the defender receives for stopping an intrusion ($R_{\mathrm{st}}> 0$), the reward for maintaining service ($R_{\mathrm{sla}}> 0$), and the loss of being intruded ($R_{\mathrm{int}} < 0$):
\begin{subequations}\label{eq:p1_reward_function}
\begin{align}
&r(\emptyset, \cdot) = 0, \label{eq:p1_reward_0}\\
&r(s, \mathsf{C}) = R_{\mathrm{sla}} + \frac{sR_{\mathrm{int}}}{L}, && s \in \{0,1\},\label{eq:p1_reward_3}\\
&r(s, \mathsf{S}) = \frac{sR_{\mathrm{st}}}{L}, && s \in \{0,1\}. \label{eq:p1_reward_5}
\end{align}
\end{subequations}
Equation (\ref{eq:p1_reward_0}) states that the reward in the terminal state is zero. Equation (\ref{eq:p1_reward_3}) states that the defender receives a positive reward ($\mathrm{R}_{\mathrm{sla}}$) for maintaining service and a loss ($\frac{R_{\mathrm{int}}}{L}$) for each time step that it is under intrusion. Lastly, (\ref{eq:p1_reward_5}) indicates that each stop incurs a cost by interrupting service (i.e., no $R_{\mathrm{sla}}$) and possibly a reward ($\frac{R_{\mathrm{st}}}{L}$) if it affects an ongoing intrusion.

\paragraph{Time horizon \texorpdfstring{$T$}{Lg}}\label{sec:p1_time_horizon}
The time horizon $T$ is a random variable that indicates the time $t > 1$ when the terminal state $\emptyset$ is reached. 

\paragraph{Theoretical analysis}
A detailed theoretical analysis of this POMDP can be found in \cite{hammar_stadler_tnsm}.
\subsection{Flow Control Markov Game}\label{sec:flow_game}
We formulate the flow control use case when the attacker follows a dynamic strategy as a partially observed zero-sum stochastic game (Markov game). The game follows similar dynamics as the POMDP defined above and has two players: the ($\mathrm{D}$)efender and the ($\mathrm{A}$)ttacker. In the following, we describe the components of the game, its evolution, and the players' objectives.

\paragraph{Time horizon $T$}
The time horizon $T > 1$ is a random variable representing the time when the attacker stops its intrusion or is prevented, depending on which event occurs first.

\paragraph{State space $\mathcal{S}$}
The game has three states: $s_t=0$ if no intrusion occurs, $s_t=1$ if an intrusion is ongoing, and $s_T=\emptyset$ if the game has ended. Hence, $\mathcal{S}= \{0,1,\emptyset\}$. The initial state is $s_1=0$. Therefore, the initial state distribution $b_1 \in \Delta(\mathcal{S})$ is the degenerate distribution $b_1(0)=1$.

\paragraph{Action spaces $\mathcal{A}_{\mathrm{k}}$}
Each player $\mathrm{k}\in \mathcal{N}$ can invoke two actions: ($\mathsf{S}$)top and ($\mathsf{C}$)ontinue. The action spaces are thus $\mathcal{A}_{\mathrm{D}}=\mathcal{A}_{\mathrm{A}}=\{\mathsf{S},\mathsf{C}\}$. Executing action $\mathsf{S}$ triggers a change in the game, while action $\mathsf{C}$ is a passive action. The attacker can invoke the stop action twice: the first to start the intrusion and the second to terminate it. The defender can invoke the stop action $L \geq 1$ times. Each invocation corresponds to a defensive action against a possible intrusion. The number of stop actions remaining to the defender at time $t$ is known to both players and is denoted by $l_t \in \{1,\hdots,L\}$. Using the encoding $(\mathsf{S},\mathsf{C})= (1,0)$, we can write $l_{t+1}=l_t-a_t^{(\mathrm{D})}$, where $a_t^{(\mathrm{D})}$ is the defender action at time $t$. At each time step, the attacker and the defender simultaneously choose their actions $a_t = (a^{(\mathrm{D})}_t, a^{(\mathrm{A})}_t)$, where $a^{(\mathrm{k})}_t \in \mathcal{A}_{\mathrm{k}}$.

\paragraph{Observation space $\mathcal{O}$}
The attacker has complete observability and knows the game state, the defender's actions, and the defender's observations. In contrast, the defender has a limited set of observations $o_t \in \mathcal{O}$, where $\mathcal{O}$ is a finite set. In our use case, $o_t$ relates to the weighted sum of IDS alerts triggered during time step $t$.

\begin{figure*}
  \centering
    \scalebox{0.85}{
      \includegraphics{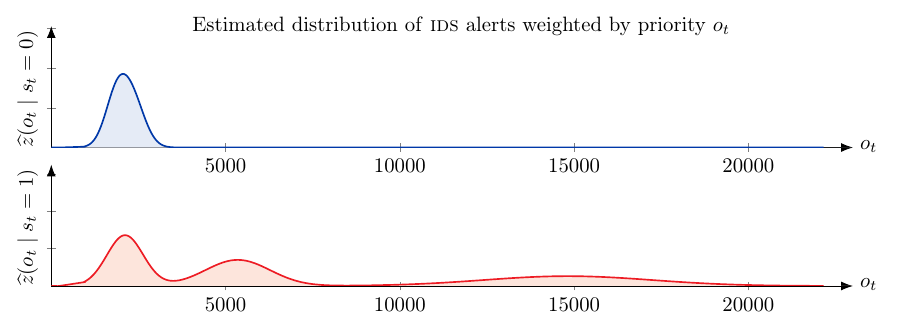}
    }
    \caption{Fitted Gaussian mixture models of $z$ when no intrusion occurs ($s_t=0$) and during intrusion ($s_t=1$) for the flow control Markov game.}
    \label{fig:p2_ids_distribution}
\end{figure*}

\paragraph{Transition function $f_l(s^{\prime} \mid s, a^{(\mathrm{D})},  a^{(\mathrm{A})})$}
At each time step $t$, a transition from $s_t$ to $s_{t+1}$ occurs with probability $f_l(s_{t+1} \mid s_t, a^{(\mathrm{D})}_t, a^{(\mathrm{A})}_t)$, where $f_l$ is defined as
\begin{subequations}\label{p2_transition_function}
\begin{align}
&f_{l>1}(0 \mid 0,\mathsf{S},\mathsf{C})= f_{l}(0 \mid 0, \mathsf{C}, \mathsf{C})=1,\label{eq:p2_tp_1}\\
&f_{l>1}(1 \mid 1, \cdot,\mathsf{C})= f_l(1 \mid 1,\mathsf{C},\mathsf{C}) = 1-\phi_{l},  \label{eq:p2_tp_2}\\
&f_{l>1}(1 \mid 0,\cdot,\mathsf{S})= f_l(1\mid 0,\mathsf{C},\mathsf{S}) = 1,\label{eq:p2_tp_3}\\
&f_{l>1}(\emptyset\mid 1, \cdot,\mathsf{C})= f_l(\emptyset \mid 1,\mathsf{C},\mathsf{C})= \phi_{l},\label{eq:p2_tp_4}\\
&f_{l=1}(\emptyset \mid \cdot,\mathsf{S},\cdot)= f_l(\emptyset \mid \emptyset,\cdot, \cdot)= f_l(\emptyset \mid 1,\cdot, \mathsf{S})= 1.\label{eq:p2_tp_7}
\end{align}
\end{subequations}
All other state transitions have probability $0$. Equations (\ref{eq:p2_tp_1})--(\ref{eq:p2_tp_2}) define the probabilities of the recurrent transitions $0\rightarrow 0$ and $1\rightarrow 1$. The game stays in state $0$ with probability $1$ if the attacker selects action $\mathsf{C}$ and $l_t-a_t^{(\mathrm{D})}>0$. Similarly, the game stays in state $1$ with probability $1-\phi_{l}$ if the attacker chooses action $\mathsf{C}$ and $l_t-a_t^{(\mathrm{D})}>0$. Here, $\phi_{l}$ denotes the probability that the defender stops the intrusion, which is a parameter of the use case. The intrusion can be stopped at any time step, either because the attacker terminates the intrusion or as a consequence of previous stop actions by the defender, i.e., the effect of a defensive action is non-immediate. We assume that $\phi_{l}$ increases with each stop action that the defender takes.

Equation (\ref{eq:p2_tp_3}) captures the transition $0 \rightarrow 1$, which occurs when the attacker chooses action $\mathsf{S}$ and $l_t-a_t^{(\mathrm{D})}>0$. (\ref{eq:p2_tp_4})--(\ref{eq:p2_tp_7}) define the probabilities of the transitions to the terminal state $\emptyset$, which is reached in three cases: (\textit{i}) when $l_t=1$ and the defender takes the final stop action $\mathsf{S}$ (i.e., when $l_t-a^{(\mathrm{D})}_t=0$); (\textit{ii}) when the intrusion is stopped by the defender with probability $\phi_{l}$; and (\textit{iii}) when $s_t=1$ and the attacker terminates the intrusion ($a^{(\mathrm{A})}_t=\mathsf{S}=1$).

\paragraph{Reward function $r_l(s, a^{(\mathrm{D})}, a^{(\mathrm{A})})$}
At time step $t$, the defender receives the reward $r_t = r_l(s_t,a^{(\mathrm{D})}_t, a^{(\mathrm{A})}_t)$ and the attacker receives the reward $-r_t$. The reward function is parameterized by the defender's reward for stopping an intrusion ($\mathrm{R}_{\mathrm{st}} > 0$), its cost of taking a defensive action ($\mathrm{R}_{\mathrm{cost}} < 0$), and its cost while an intrusion occurs ($\mathrm{R}_{\mathrm{int}} < 0$), as defined below.
\begin{subequations}\label{eq:p2_reward_function}
\begin{align}
&r_l(\emptyset, \cdot)=0,\label{eq:p2_reward_0}\\
&r_l(1, \cdot,\mathsf{S}) = 0, \label{eq:p2_reward_1}\\
&r_l(0, \mathsf{C},\cdot)=0, \label{eq:p2_reward_2}\\
&r_l(0, \mathsf{S},\cdot)=\frac{\mathrm{R}_{\mathrm{cost}}}{l_t}, && l_t \in \{1,2,\hdots,L\},   \label{eq:p2_reward_3}\\
&r_l(1, \mathsf{S},\mathsf{C})=\frac{\mathrm{R}_{\mathrm{st}}}{l_t}, && l_t \in \{1,2,\hdots,L\}, \label{eq:p2_reward_4}\\
&r_l(1, \mathsf{C},\mathsf{C})=\mathrm{R}_{\mathrm{int}}.  \label{eq:p2_reward_5}
\end{align}
\end{subequations}
Equations (\ref{eq:p2_reward_0})--(\ref{eq:p2_reward_1}) state that the reward is zero in the terminal state and when the attacker terminates an intrusion. Equation (\ref{eq:p2_reward_2}) states that the defender incurs no cost when no attack occurs and it does not take a defensive action. Equation (\ref{eq:p2_reward_3}) indicates that the defender incurs a cost when taking a defensive action if no intrusion is ongoing. Equation (\ref{eq:p2_reward_4}) states that the defender receives a reward when taking a stop action while an intrusion occurs. Lastly, (\ref{eq:p2_reward_5}) indicates that the defender incurs a cost for each time step during which an intrusion occurs.

\paragraph{Observation function $z$} We estimate the observation function $z$ based on data from the digital twin. Specifically, at the end of every time step in the digital twin, i.e., at the end of each $30$s interval, we collect the number of IDS alerts with priorities $1$--$4$ that occurred during the time step, where priorities $1$--$4$ refer to the Snort priorities ``very low'', ``low'', ``medium'', and ``high'', respectively \cite{snort}\footnote{Note that according to Snort's terminology \cite{snort}, $1$ is the highest priority. We invert the labeling in our model for convenience.}. We do so for $23,000$ time steps, which provides us with a dataset to estimate the distribution of IDS alerts. Using this dataset, we apply expectation-maximization \cite{em_demp_77} to fit Gaussian mixture distributions $\widehat{z}(\cdot \mid 0)$ and $\widehat{z}(\cdot \mid 1)$ as estimates of $z(\cdot \mid 0)$ and $z(\cdot \mid 1)$, which represent the true observation distributions in the target infrastructure.

Figure~\ref{fig:p2_ids_distribution} shows the fitted models over the discrete observation space $\mathcal{O} = \{0,1,\hdots,22000\}$. We note that $\widehat{z}(\cdot \mid 0)$ and $\widehat{z}(\cdot \mid 1)$ are (discretized) Gaussian mixtures with one and three components, respectively. Both mixtures have the most probability mass within $0$--$5000$. The distribution $\widehat{z}(\cdot \mid 1)$ also has substantial probability mass at larger values.

\paragraph{Theoretical analysis}
A detailed theoretical analysis of this game can be found in \cite{hammar_stadler_tnsm_game_23}.
\subsection{Network Segmentation Markov Game}\label{sec:segmentation_game}
We formulate the network segmentation use case when the attacker follows a dynamic strategy as a partially observed zero-sum stochastic game (Markov game). The game has two players: the ($\mathrm{D}$)efender and the ($\mathrm{A}$)ttacker.

The game is played on an IT infrastructure with application servers connected by a communication network that is segmented into zones. Overlaid on this physical infrastructure is a virtual infrastructure with a tree structure that includes \textit{nodes}, which collectively offer services to clients. A service is modeled as a \textit{workflow}, which comprises a set of interdependent nodes. A dependency between two nodes reflects information exchange through service invocations. 

In the following, we describe the components of the game, its evolution, and the players' objectives.
\paragraph{Infrastructure}
We model the virtual infrastructure as a (finite) directed graph $\mathcal{G} = \langle \{\mathrm{gw}\}\cup \mathcal{V}, \mathcal{E} \rangle$. The graph has a tree structure rooted at the gateway $\mathrm{gw}$. Each node $i \in \mathcal{V}$ has three state variables. Variable $v_{i,t}^{(\mathrm{R})}$ represents the reconnaissance state. We have $v_{i,t}^{(\mathrm{R})}=1$ if the attacker has discovered the node, $0$ otherwise. Variable $v_{i,t}^{(\mathrm{I})}$ represents the intrusion state. We have $v_{i,t}^{(\mathrm{I})}=1$ if the attacker has compromised the node, $0$ otherwise. Lastly, variable $v^{(\mathrm{Z})}_{i,t}$ indicates the zone in which the node resides. We call a node \textit{active} if it is functional as part of a workflow (denoted $\alpha_{i,t}=1$). Due to a defender action (e.g., a shutdown), a node $i \in \mathcal{V}$ may become inactive ($\alpha_{i,t}=0$). The active state is determined by $v^{(\mathrm{Z})}_{i,t}$, i.e., $\alpha_{i,t}$ is a function of $v^{(\mathrm{Z})}_{i,t}$.

\paragraph{Workflows}

We model a workflow $w \in \mathcal{W}$ as a subtree $\mathcal{G}_{w} = \langle \{\mathrm{gw}\}\cup \mathcal{V}_{w}, \mathcal{E}_{w} \rangle$ of the infrastructure graph. Workflows do not overlap except for the gateway, which belongs to all workflows.

\paragraph{Attacker}
At each time $t$, the attacker takes an action $a^{(\mathrm{A})}_t$, which is defined as the composition of the local actions on all nodes $a^{(\mathrm{A})}_t = (a^{(\mathrm{A})}_{1,t}, \hdots, a^{(\mathrm{A})}_{|\mathcal{V}|,t}) \in \mathcal{A}_{\mathrm{A}}$, where the set $\mathcal{A}_{\mathrm{A}}$ is finite. A local action is either a null action (denoted with $\bot$) or an offensive action. An offensive action on a node $i$ may change the reconnaissance state $v^{(\mathrm{R})}_{i,t}$ or the intrusion state $v^{(\mathrm{I})}_{i,t}$. A node $i$ can only be compromised if it is discovered, i.e., if $v^{(\mathrm{R})}_{i,t}=1$. We express this constraint as $a^{(\mathrm{A})}_t \in \mathcal{A}_{\mathrm{A}}(s^{(\mathrm{A})}_t)$.

The attacker state $s^{(\mathrm{A})}_t = \big(v^{(\mathrm{R})}_{i,t},v^{(\mathrm{I})}_{i,t}\big)_{i \in \mathcal{V}} \in \mathcal{S}_{\mathrm{A}}$ evolves as
\begin{align}
s^{(\mathrm{A})}_{t+1} \sim f_{\mathrm{A}}\big(\cdot \mid s^{(\mathrm{A})}_t, a^{(\mathrm{A})}_t,a^{(\mathrm{D})}_t\big),\label{eq:p3_transition_atc}
\end{align}
where $a^{(\mathrm{D})}_t$ represents the defender action at time $t$, as defined below.

\paragraph{Defender}
At each time $t$, the defender takes action $a^{(\mathrm{D})}_t$, which is defined as the composition of the local actions on all nodes $a^{(\mathrm{D})}_t = (a^{(\mathrm{D})}_{1,t}, \hdots, a^{(\mathrm{D})}_{|\mathcal{V}|,t}) \in \mathcal{A}_{\mathrm{D}}$, where the set $\mathcal{A}_{\mathrm{D}}$ is finite. A local action is either a defensive action or a null action $\bot$. Each defensive action $a^{(\mathrm{D})}_{i,t}\neq \bot$ leads to $s^{(\mathrm{A})}_{i,t+1}=(0,0)$ and may affect $v^{(\mathrm{Z})}_{i,t+1}$.

The defender state $s^{(\mathrm{D})}_t = \big(v^{(\mathrm{Z})}_{i,t}\big)_{i \in \mathcal{V}} \in \mathcal{S}_{\mathrm{D}}$ evolves as
\begin{align}
s^{(\mathrm{D})}_{t+1} \sim f_{\mathrm{D}}\big(\cdot \mid s^{(\mathrm{D})}_t, a^{(\mathrm{D})}_t\big).\label{eq:p3_transition_def}
\end{align}
\paragraph{Clients}

Clients consume infrastructure services by accessing workflows. We model client behavior through stationary stochastic processes, which affect the observations available to the attacker and the defender. That is, the clients are implicitly modeled by the observation function $z$, as defined next.

\paragraph{Observability and strategies}
At each time $t$, the defender and the attacker both observe $o_t = \left(o_{1,t},\hdots,o_{|\mathcal{V}|,t}\right) \in \mathcal{O}$, where $\mathcal{O}$ is finite\footnote{In our use case,  $o_{i,t}$ relates to the number of intrusion alerts associated with node $i$.}. The observation $o_t$ is drawn from the random vector $O_t = (O_{1,t},\hdots,O_{|\mathcal{V}|,t})$ whose marginal distributions $z_{O_1}, \hdots, z_{O_{|\mathcal{V}|}}$ are stationary and conditionally independent given $s_{i,t}=(s^{(\mathrm{D})}_{i,t},s^{(\mathrm{A})}_{i,t})$. (Note that $z_{O_i}$ depends on the traffic generated by clients.) As a consequence, the joint conditional distribution $z(o\mid s)$ is given by
\begin{align}
z\big(o \mid s\big)=\prod^{|\mathcal{V}|}_{i=1} z_{O_i}\big(o_{i} \mid s_{i}\big) && \forall o \in \mathcal{O}, s \in \mathcal{S}_{\mathrm{A}} \times \mathcal{S}_{\mathrm{D}}.\label{eq:p3_obs_1}
\end{align}
The sequence of observations and states at times $1,\hdots,t$ forms the histories
\begin{align*}
h^{(\mathrm{D})}_t&= (b^{(\mathrm{D})}_1, s^{(\mathrm{D})}_1,a_1^{(\mathrm{D})},o_2,\hdots,a_{t-1}^{(\mathrm{D})},s^{(\mathrm{D})}_t,o_t) \in \mathcal{H}_{\mathrm{D}},\\
h^{(\mathrm{A})}_t &= (b^{(\mathrm{A})}_1, s_1^{(\mathrm{A})},a_1^{(\mathrm{A})},o_2,\hdots,a_{t-1}^{(\mathrm{A})},s_t^{(\mathrm{A})},o_t) \in \mathcal{H}_{\mathrm{A}},
\end{align*}
where $s^{(\mathrm{D})} \sim b^{(\mathrm{D})}_1$ and $s^{(\mathrm{A})} \sim b^{(\mathrm{A})}_1$ are the initial state distributions.

Based on their respective histories, the defender and the attacker select actions according to their strategies. The defender's behavior strategy is defined as $\pi_{\mathrm{D}}\in \Pi_{\mathrm{D}} = \mathcal{H}_{\mathrm{D}} \rightarrow \Delta(\mathcal{A}_{\mathrm{D}})$ and the attacker's behavior strategy is defined as $\pi_{\mathrm{A}} \in \Pi_{\mathrm{A}} = \mathcal{H}_{\mathrm{A}} \rightarrow \Delta(\mathcal{A}_{\mathrm{A}})$.

\paragraph{Defender Objective}
When selecting the strategy $\pi_{\mathrm{D}}$, the defender must balance two conflicting objectives: maximizing the workflow utility towards its clients and minimizing the cost of intrusion. The weight $\eta \geq 0$ controls the trade-off between these two objectives, which results in the bi-objective
\begin{align}
  &J = \sum_{t=1}^{\infty}\gamma^{t-1} \left(\sum_{w \in \mathcal{W}}\sum_{i \in \mathcal{V}_{w}}\underbrace{\eta u^{(\mathrm{W})}_{i,t}}_{\text{workflow utility}} - \underbrace{c_{i,t}^{(\mathrm{I})}}_{\text{intrusion cost}}\right), \label{eq:p3_objective_fun}
\end{align}
where $\gamma \in [0,1)$ is a discount factor, $c_{i,t}^{(\mathrm{I})}= v^{(\mathrm{I})}_{i,t} + c^{(\mathrm{A})}(a^{(\mathrm{D})}_{i,t})$ is the intrusion cost associated with node $i$ at time $t$\footnote{$c^{(\mathrm{A})}$ is a non-negative function that represents the operational costs of defender actions.},  and $u^{(\mathrm{W})}_{i,t}$ expresses the workflow utility associated with node $i$ at time $t$.

\paragraph{Markov Game}
When the game starts at $t=1$, $s^{(\mathrm{D})}_1$ and $s_1^{(\mathrm{A})}$ are sampled from $b^{(\mathrm{D})}_1$ and $b^{(\mathrm{A})}_1$, respectively. A play of the game proceeds in time steps $t=1,2,\hdots$. At each time $t$, the defender observes $h^{(\mathrm{D})}_t$ and the attacker observes $h^{(\mathrm{A})}_t$. Based on these histories, both players select actions according to their respective strategies, i.e., $a^{(\mathrm{D})}_{t} \sim \pi_{\mathrm{D}}(\cdot \mid h^{(\mathrm{D})}_t)$ and $a^{(\mathrm{A})}_{t} \sim \pi_{\mathrm{A}}(\cdot \mid h^{(\mathrm{A})}_t)$. As a result of these actions, five events occur at time $t+1$: (\textit{i}) $o_{t+1}$ is sampled from $z$; (\textit{ii}) $s_{t+1}^{(\mathrm{D})}$ is sampled from $f_{\mathrm{D}}$; (\textit{iii}) $s_{t+1}^{(\mathrm{A})}$ is sampled from $f_{\mathrm{A}}$; (\textit{iv}) the defender receives the reward $r(s_t,a^{(\mathrm{D})}_{t})$; and (\textit{v}) the attacker receives the reward $-r(s_t,a^{(\mathrm{D})}_{t})$, where the reward function $r$ is defined by the expression within brackets in (\ref{eq:p3_objective_fun}).
\paragraph{Theoretical analysis}
A detailed theoretical analysis of this game can be found in \cite{hammar_gamesec23}.
\subsection{Replication Control MDP}\label{sec:replication_mdp}
We formulate the replication control use case when the attacker follows a static strategy as the following MDP.

\paragraph{States}
We define the state $s_{t}$ to represent the expected number of healthy replicas at time $t$. The state space is $\mathcal{S} = \{0,1,\hdots,s_{\mathrm{max}}\}$ and the initial state is $s_{1}=N_1$.

\paragraph{Actions}
The action $a_t = 1$ means that a new replica is added to the system; $a_t=0$ otherwise. Hence, the action space is $a_t\in \{0,1\} = \mathcal{A}$.

\paragraph{Transition probabilities}
The state $s_t$ evolves as
\begin{align}
s_{t+1} &\sim f(\cdot \mid s_t, a_t).\label{eq:transition_system_controller}
\end{align}
We estimate the conditional probability distribution $f$ based on measurements from the digital twin. A subset of the estimated conditional distributions is illustrated in Fig.~\ref{fig:pmf_1}.

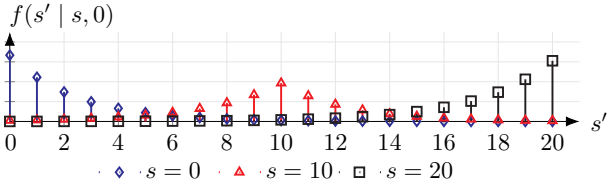
\begin{figure}[H]
  \centering
    \scalebox{0.91}{
      \begin{tikzpicture}[
    dot/.style={
        draw=black,
        fill=blue!90,
        circle,
        minimum size=3pt,
        inner sep=0pt,
        solid,
    },
    ]

\node[scale=1] (kth_cr) at (0,2.15)
{
  \begin{tikzpicture}[declare function={sigma(\x)=1/(1+exp(-\x));
      sigmap(\x)=sigma(\x)*(1-sigma(\x));}]
\pgfplotstableread{
0 0.333451
1 0.222301
2 0.148200
3 0.098800
4 0.065867
5 0.043911
6 0.029274
7 0.019516
8 0.013010
9 0.008673
10 0.005781
11 0.003853
12 0.002567
13 0.001710
14 0.001137
15 0.000754
16 0.000497
17 0.000323
18 0.000203
19 0.000118
20 0.000054
}\datatablee
\pgfplotstableread{
0 0.005066
1 0.007290
2 0.010491
3 0.015098
4 0.021726
5 0.031265
6 0.044991
7 0.064743
8 0.093167
9 0.134069
10 0.192929
11 0.128619
12 0.085746
13 0.057164
14 0.038109
15 0.025406
16 0.016937
17 0.011291
18 0.007527
19 0.005018
20 0.003345
}\datatableee
\pgfplotstableread{
0 0.000113
1 0.000238
2 0.000393
3 0.000599
4 0.000884
5 0.001287
6 0.001861
7 0.002685
8 0.003868
9 0.005569
10 0.008016
11 0.011537
12 0.016603
13 0.023893
14 0.034383
15 0.049478
16 0.071200
17 0.102459
18 0.147441
19 0.212172
20 0.305320
}\datatableeee
\begin{axis}[
        xmin=0,
        xmax=21,
        ymin=0,
        ymax=0.45,
        width =1.2\columnwidth,
        height = 0.35\columnwidth,
        axis lines=center,
        xmajorgrids=true,
        ymajorgrids=true,
        major grid style = {lightgray},
        minor grid style = {lightgray!25},
        ytick={0.0, 0.1,0.2, 0.3, 0.4, 0.5, 0.6, 0.7, 0.8, 0.9, 1.0},
        yticklabels={$0$, $$, $$, $$, $$, $$, $$, $$, $$, $$,$$},
        xtick={0,2,4,6,8,10,12,14,16,18,20},
        xticklabels={$0$, $2$, $4$, $6$, $8$, $10$, $12$, $14$, $16$, $18$, $20$},
        scaled y ticks=false,
        yticklabel style={
        /pgf/number format/fixed,
        /pgf/number format/precision=5
        },
        xlabel style={below right},
        ylabel style={above left},
        axis line style={-{Latex[length=2mm]}},
        legend style={at={(0.8,-0.35)}},
        legend columns=3,
        legend style={
          draw=none,
            /tikz/column 2/.style={
                column sep=5pt,
              }
              }
              ]
              \addplot+[ycomb,Blue,thick, mark=diamond] table [x index=0, y index=1] {\datatablee};
              \addplot+[ycomb,Red,thick, mark=triangle] table [x index=0, y index=1] {\datatableee};
              \addplot+[ycomb,Black,thick, mark=square] table [x index=0, y index=1] {\datatableeee};
\legend{$s=0$, $s=10$, $s=20$}
\end{axis}
\end{tikzpicture}
};

\node[inner sep=0pt,align=center, scale=1, rotate=0, opacity=1] (obs) at (-3.95,1.7)
{
  $0$
};

\node[inner sep=0pt,align=center, scale=1, rotate=0, opacity=1] (obs) at (4.65,2)
{
  $s^{\prime}$
};
\node[inner sep=0pt,align=center, scale=1, rotate=0, opacity=1] (obs) at (-3.2,3.55)
{
  $f(s^{\prime} \mid s, 0)$
};

\end{tikzpicture}
    }
    \caption{Transition function for the replication control MDP.}
    \label{fig:pmf_1}
  \end{figure}

\paragraph{Objective}
Increasing the replication factor $s_t$ improves service availability $T^{(\mathrm{A})}$ but increases cost. ($T^{(\mathrm{A})}$ is the fraction of time steps where service is available.) The goal of the controller is thus to find the optimal cost-redundancy trade-off, i.e., to minimize
\begin{align}
J&= \lim_{T \rightarrow \infty}\left[\sum_{t=1}^T\frac{a_t}{T}\right],&& \text{subject to }T^{(\mathrm{A})} \geq \epsilon_{\mathrm{A}},\label{eq:p4_objective_response}
\end{align}
where $\epsilon_{\mathrm{A}}$ is the chosen lower bound on service availability. For instance, if $\epsilon_{\mathrm{A}}=0.999$, then at most $8.4$ hours of service disruption per year is allowed.

\paragraph{Theoretical analysis}
A detailed theoretical analysis of this MDP can be found in \cite{dsn24_hammar_stadler}.
\subsection{Replication Control Markov Game}\label{sec:replication_game}
The Markov game model of the replication control use case follows the same model as the MDP, except that the attacker's strategy is dynamic. Specifically, the attacker can control which nodes in the system to attack, which causes them to be compromised with probability $p_{\mathrm{A}}=0.01$. Hence, the attacker's strategy influences the transition function in \eqref{eq:transition_system_controller}. The objective of the attacker is diametrically opposed to the controller, i.e., it is a zero-sum game.

\paragraph{Theoretical analysis}
A detailed theoretical analysis of this game can be found in \cite{hammar_gamesec24}.

\subsection{Recovery Control POMDP}\label{sec:recovery_control}
The recovery control use case involves a networked system with $K$ service replicas. We formulate this use case as the following POMDP.

\paragraph{States}
Each replica has two states: $1$ (compromised) or $0$ (safe), i.e., $s=(s^1,\hdots,s^K)$ where $s^l \in \{0,1\}$. Compromises occur randomly over time and incur operational costs. The transition probabilities are defined as follows. If replica $l$ is compromised ($s^l=1$), then it remains so until recovery is applied ($a^l=1$), at which point the state $s^l$ is set to $0$. Otherwise, the probability that it becomes compromised is $\min\{0.2(1+\mathcal{N}_l(s)), 1\}$, where $\mathcal{N}_l(s)$ is the number of compromised neighbors of replica $l$ in the network.

\paragraph{Actions}
An action is defined as a vector $a=(a^1,\hdots,a^K)$, where each $a^l$ determines whether to recover component $l$ ($a^l = 1$) or take no action ($a^l = 0$). The goal is to determine an optimal recovery strategy $\pi^\star$ that balances security requirements against recovery costs.

\paragraph{Observations}
Intrusion detection systems generate security alerts $o = (o^1, \ldots, o^K)$ that provide partial indications of the replicas' states. We define the observation distribution as
\begin{align*}
p(o \mid s, a) &= \prod_{l=1}^Kp(o^l \mid s^l), && \text{for all }o\in \mathcal{O},s\in \mathcal{S},a\in \mathcal{A},
\end{align*}  
where each $p(o^l \mid s^l)$ is estimated based on measurements from the digital twin. A subset of the estimated distributions is shown in Fig.~\ref{fig:obs_dist}.

\begin{figure*}
  \centering
\scalebox{0.7}{
    \begin{tikzpicture}

\pgfplotstableread{
25.0 1.6003200640128026e-05
75.0 2.8005601120224046e-05
125.0 0.00015603120624124825
175.0 0.0004600920184036807
225.0 0.0010682136427285457
275.0 0.002152430486097219
325.0 0.003240648129625925
375.0 0.0040888177635527104
425.0 0.003644728945789158
475.0 0.002448489697939588
525.0 0.0016323264652930586
575.0 0.0007161432286457291
625.0 0.00023604720944188836
675.0 7.201440288057611e-05
725.0 2.8005601120224046e-05
775.0 1.2002400480096018e-05
825.0 0.0
875.0 0.0
925.0 0.0
975.0 00.
}\nointzero

\pgfplotstableread{
25.0 4.022526146419952e-06
75.0 1.6090104585679808e-05
125.0 1.6090104585679808e-05
175.0 0.00012067578439259854
225.0 0.00018101367658889783
275.0 0.0004304102976669349
325.0 0.0006838294448913918
375.0 0.0010981496379726468
425.0 0.0016452131938857603
475.0 0.001830249396621078
525.0 0.0023813354786806113
575.0 0.0024094931617055513
625.0 0.0025382139983909893
675.0 0.0021158487530168945
725.0 0.0017940466613032985
775.0 0.0010579243765084473
825.0 0.0008487530168946098
875.0 0.0004545454545454545
925.0 0.00026548672566371683
975.0 0.00010860820595333871
}\intzero

\pgfplotstableread{
25.0 0.000062
75.0 0.000281
125.0 0.000955
175.0 0.002011
225.0 0.002550
275.0 0.003115
325.0 0.003551
375.0 0.003250
425.0 0.002581
475.0 0.002011
525.0 0.001955
575.0 0.001502
625.0 0.000652
675.0 0.000195
725.0 0.000048
775.0 0.000011
825.0 0.000002
875.0 0.000000
925.0 0.000000
975.0 0.000000
}\nointone

\pgfplotstableread{
25.0 0.000002
75.0 0.000008
125.0 0.000025
175.0 0.000075
225.0 0.000198
275.0 0.000455
325.0 0.000912
375.0 0.001651
425.0 0.002488
475.0 0.002815
525.0 0.002951
575.0 0.002988
625.0 0.002811
675.0 0.002515
725.0 0.002301
775.0 0.001655
825.0 0.000851
875.0 0.000321
925.0 0.000105
975.0 0.000031
}\intone

\pgfplotstableread{
25.0 0.000105
75.0 0.000288
125.0 0.000672
175.0 0.001331
225.0 0.002235
275.0 0.003188
325.0 0.003881
375.0 0.003950
425.0 0.003310
475.0 0.002251
525.0 0.001475
575.0 0.001221
625.0 0.001095
675.0 0.000550
725.0 0.000168
775.0 0.000045
825.0 0.000010
875.0 0.000002
925.0 0.000000
975.0 0.000000
}\nointtwo

\pgfplotstableread{
25.0 0.000011
75.0 0.000030
125.0 0.000075
175.0 0.000161
225.0 0.000318
275.0 0.000572
325.0 0.000941
375.0 0.001420
425.0 0.001965
475.0 0.002450
525.0 0.002758
575.0 0.002791
625.0 0.002510
675.0 0.002181
725.0 0.002015
775.0 0.001550
825.0 0.000678
875.0 0.000201
925.0 0.000051
975.0 0.000012
}\inttwo

\pgfplotstableread{
25.0 0.000109
75.0 0.000224
125.0 0.000421
175.0 0.000720
225.0 0.001103
275.0 0.001478
325.0 0.001798
375.0 0.001995
425.0 0.002046
475.0 0.001925
525.0 0.001683
575.0 0.001402
625.0 0.001224
675.0 0.001301
725.0 0.001795
775.0 0.001251
825.0 0.000450
875.0 0.000089
925.0 0.000015
975.0 0.000002
}\nointthree

\pgfplotstableread{
25.0 0.000095
75.0 0.000199
125.0 0.000382
175.0 0.000659
225.0 0.001021
275.0 0.001395
325.0 0.001691
375.0 0.001890
425.0 0.001962
475.0 0.001889
525.0 0.001701
575.0 0.001449
625.0 0.001211
675.0 0.001155
725.0 0.001391
775.0 0.001685
825.0 0.001015
875.0 0.000321
925.0 0.000065
975.0 0.000010
}\intthree

\pgfplotsset{/dummy/workaround/.style={/pgfplots/axis on top}}

\node[scale=1] (kth_cr) at (3.78,-0.58)
{
\begin{tikzpicture}
  \begin{axis}[
      width=7cm,
      height=3.5cm,
      ytick=\empty,
      ybar, 
      bar width=45,
      xmin=0,
      xmax=1000,
      ymin=0,
      ymax=0.005,
        legend style={at={(2.7,-0.65)}, nodes={scale=0.9, transform shape,anchor=west}},
        legend columns=0,
        legend style={
          draw=none,
      anchor=north east,
      /tikz/every even column/.style={anchor=west, column sep=5pt},
      /tikz/every odd column/.style={anchor=west},          
            /tikz/column 2/.style={
                column sep=5pt,
              }
              },      
              axis lines*=left,
              axis line style={-{Latex[length=1.8mm]}},
      enlarge x limits=0.05,
      yticklabel style={
          /pgf/number format/fixed,
          /pgf/number format/precision=4
      },
      ]
    \addplot[fill=bluefour, draw=black] table[x index=0, y index=1] {\nointzero};
    \addplot[fill=red, draw=black, postaction={pattern=dots}] table[x index=0, y index=1] {\intzero};
    \legend{Safe ($s^i=0$), Compromised ($s^i=1$)}
  \end{axis}
\end{tikzpicture}
};

\node[scale=1] (kth_cr) at (5.15,0)
{
\begin{tikzpicture}
  \begin{axis}[
      width=7cm,
      height=3.5cm,
      ytick=\empty,
      ybar, 
      bar width=45,
      xmin=0,
      xmax=1000,
      ymin=0,
      ymax=0.005,
      axis lines*=left,
      axis line style={-{Latex[length=1.8mm]}},
      enlarge x limits=0.05,
      yticklabel style={
          /pgf/number format/fixed,
          /pgf/number format/precision=4
      },
      ]
    \addplot[fill=bluefour, draw=black] table[x index=0, y index=1] {\nointone};
    \addplot[fill=red, draw=black, postaction={pattern=dots}] table[x index=0, y index=1] {\intone};
  \end{axis}
\end{tikzpicture}
};

\node[scale=1] (kth_cr) at (11,0)
{
\begin{tikzpicture}
  \begin{axis}[
      width=7cm,
      height=3.5cm,
      ytick=\empty,
      ybar, 
      bar width=45,
      xmin=0,
      xmax=1000,
      ymin=0,
      ymax=0.005,
      axis lines*=left,
      axis line style={-{Latex[length=1.8mm]}},
      enlarge x limits=0.05,
      yticklabel style={
          /pgf/number format/fixed,
          /pgf/number format/precision=4
      },
      ]
    \addplot[fill=bluefour, draw=black] table[x index=0, y index=1] {\nointtwo};
    \addplot[fill=red, draw=black, postaction={pattern=dots}] table[x index=0, y index=1] {\inttwo};
  \end{axis}
\end{tikzpicture}
};

\node[scale=1] (kth_cr) at (16.8,0)
{
\begin{tikzpicture}
  \begin{axis}[
      width=7cm,
      height=3.5cm,
      ytick=\empty,
      ybar, 
      bar width=45,
      xmin=0,
      xmax=1000,
      ymin=0,
      ymax=0.005,
      axis lines*=left,
      axis line style={-{Latex[length=1.8mm]}},
      enlarge x limits=0.05,
      yticklabel style={
          /pgf/number format/fixed,
          /pgf/number format/precision=4
      },
      ]
    \addplot[fill=bluefour, draw=black] table[x index=0, y index=1] {\nointthree};
    \addplot[fill=red, draw=black, postaction={pattern=dots}] table[x index=0, y index=1] {\intthree};
  \end{axis}
\end{tikzpicture}
};


\node[inner sep=0pt,align=center, scale=0.9, rotate=0, opacity=1] (obs) at (-0.45,-1.4)
{
  Number of failure alerts $o^1$
};
\node[inner sep=0pt,align=center, scale=0.9, rotate=0, opacity=1] (obs) at (5.25,-1.4)
{
  Number of failure alerts $o^2$
};
\node[inner sep=0pt,align=center, scale=0.9, rotate=0, opacity=1] (obs) at (11.1,-1.4)
{
  Number of failure alerts $o^3$
};
\node[inner sep=0pt,align=center, scale=0.9, rotate=0, opacity=1] (obs) at (16.7,-1.4)
{
  Number of failure alerts $o^4$
};

\node[inner sep=0pt,align=center, scale=0.9, rotate=0, opacity=1] (obs) at (-2.63,1.35)
{
  Probability
};

\node[inner sep=0pt,align=center, scale=0.9, rotate=0, opacity=1] (obs) at (3.2,1.35)
{
  Probability
};
\node[inner sep=0pt,align=center, scale=0.9, rotate=0, opacity=1] (obs) at (9.05,1.35)
{
  Probability
};
\node[inner sep=0pt,align=center, scale=0.9, rotate=0, opacity=1] (obs) at (14.87,1.35)
{
  Probability
};
\end{tikzpicture}
  }
\caption{Empirical observation distributions for the recovery control use case based on measurements from the digital twin.}\label{fig:obs_dist}
\end{figure*}
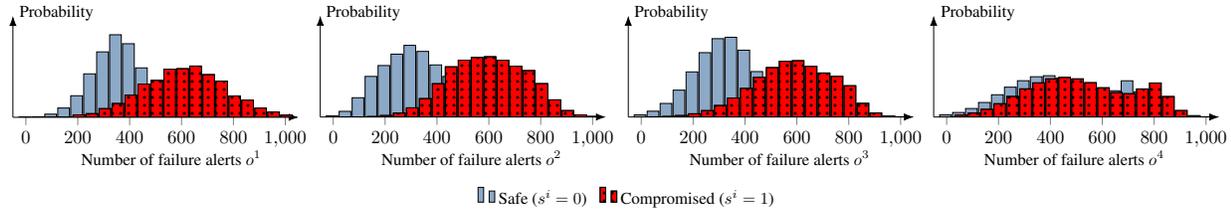

\paragraph{Rewards}
We define the reward function as
\begin{align}
r(s, a) = -\sum_{l=1}^{K}\left(\overbrace{2s^l(1-a^l)}^{\text{intrusion cost}} + \overbrace{a^l(1-s^l)}^{\text{recovery cost}}\right),\label{example_pomdp_cost}
\end{align}
i.e., negative rewards are incurred for unmitigated intrusions ($s^l=1$) and unnecessary recovery actions ($a^l=1$ and $s^l=0$).
\paragraph{Theoretical analysis}
A detailed theoretical analysis of this POMDP can be found in \cite{tifs_25_HLALB}.
\section{Proof of Proposition \ref{prop:misspecification_bound}}\label{app:misspecification_bound}
The proof follows the same chain of reasoning as the proof of the simulation lemma in \cite{Kearns2002}. For notational brevity, we use $\bm{\pi}(s)$ as a shorthand for $(\pi_{\mathrm{D}}(s), \pi_{\mathrm{A}}(s))$. We start by expanding the difference $|\tilde{J}_{\bm{\pi}}(s)-J_{\bm{\pi}}(s)|$ as
\begin{align*}
  &|\tilde{J}_{\bm{\pi}}(s)-J_{\bm{\pi}}(s)| \\
  &= \Biggl| r(s, \bm{\pi}(s)) + \gamma\sum_{s^{\prime} \in \mathcal{S}} \tilde{f}(s^{\prime} \mid s,\bm{\pi}(s))\tilde{J}_{\bm{\pi}}(s^{\prime})  - \biggl(\\
&\quad\quad r(s,\bm{\pi}(s)) + \gamma\sum_{s^{\prime} \in \mathcal{S}} f(s^{\prime} \mid s,\bm{\pi}(s))J_{\bm{\pi}}(s^{\prime})\biggr)\Biggr|  \\
  &=  \Biggl| \gamma\sum_{s^{\prime} \in \mathcal{S}} \tilde{f}(s^{\prime} \mid s,\bm{\pi}(s))\tilde{J}_{\bm{\pi}}(s^{\prime}) - \\
  & \quad\quad  \gamma\sum_{s^{\prime} \in \mathcal{S}} f(s^{\prime} \mid s,\bm{\pi}(s))J_{\bm{\pi}}(s^{\prime})\Biggr|  \\
  &=  \Biggl| \gamma\sum_{s^{\prime} \in \mathcal{S}} \tilde{f}(s^{\prime} \mid s,\bm{\pi}(s))\tilde{J}_{\bm{\pi}}(s^{\prime})  -\\
  &\quad\quad \gamma\sum_{s^{\prime} \in \mathcal{S}} f(s^{\prime} \mid s,\bm{\pi}(s))J_{\bm{\pi}}(s^{\prime}) +\\
  &\quad\quad \gamma\sum_{s^{\prime} \in \mathcal{S}} \tilde{f}(s^{\prime} \mid s,\bm{\pi}(s))J_{\bm{\pi}}(s^{\prime}) -\\
  &\quad\quad\gamma\sum_{s^{\prime} \in \mathcal{S}} \tilde{f}(s^{\prime} \mid s,\bm{\pi}(s))J_{\bm{\pi}}(s^{\prime})\Biggr|\\
&=  \Biggl| \gamma\sum_{s^{\prime} \in \mathcal{S}} \tilde{f}(s^{\prime} \mid s,\bm{\pi}(s))\left(\tilde{J}_{\bm{\pi}}(s^{\prime})-J_{\bm{\pi}}(s^{\prime})\right)  + \\
&\quad \gamma\sum_{s^{\prime} \in \mathcal{S}} \left(\tilde{f}(s^{\prime} \mid s,\bm{\pi}(s))-f(s^{\prime} \mid s,\bm{\pi}(s))\right)J_{\bm{\pi}}(s^{\prime})\Biggr|\\
  &\leq \gamma\norm{\tilde{J}_{\bm{\pi}}-J_{\bm{\pi}}}_{\infty} + \\
  &\quad \quad\Biggl|\gamma\sum_{s^{\prime} \in \mathcal{S}} \left(\tilde{f}(s^{\prime} \mid s,\bm{\pi}(s))-f(s^{\prime} \mid s,\bm{\pi}(s))\right)J_{\bm{\pi}}(s^{\prime})\Biggr|\\
  &\numleq{a} \gamma\norm{\tilde{J}_{\bm{\pi}}-J_{\bm{\pi}}}_{\infty} + \\
  &\quad\quad \gamma\sum_{s^{\prime} \in \mathcal{S}} \biggl|\left(\tilde{f}(s^{\prime} \mid s,\bm{\pi}(s))-f(s^{\prime} \mid s,\bm{\pi}(s))\right)\biggr|\frac{\beta}{1-\gamma}  \\
&\leq \gamma\norm{\tilde{J}_{\bm{\pi}}-J_{\bm{\pi}}}_{\infty} + \frac{\gamma \alpha \beta}{1-\gamma},
\end{align*}
where (a) follows because $|J_{\bm{\pi}}(s)| \leq \sum_{t=0}^{\infty}\gamma^{t}\beta = \frac{\beta}{1-\gamma}$ and the fact that $|ab|=|a||b|$ (we use the triangle inequality to move the absolute value inside the sum). Since this upper bound holds for any state $s$, we have
\begin{align*}
&\norm{\tilde{J}_{\bm{\pi}}-J_{\bm{\pi}}}_{\infty} \leq \gamma\norm{\tilde{J}_{\bm{\pi}}-J_{\bm{\pi}}}_{\infty} + \frac{\gamma \alpha \beta}{1-\gamma}\\
&\implies \norm{\tilde{J}_{\bm{\pi}}-J_{\bm{\pi}}}_{\infty}-\gamma\norm{\tilde{J}_{\bm{\pi}}-J_{\bm{\pi}}}_{\infty}\leq \frac{\gamma \alpha \beta}{1-\gamma}\\
&\implies (1-\gamma)\norm{\tilde{J}_{\bm{\pi}}-J_{\bm{\pi}}}_{\infty}\leq \frac{\gamma \alpha \beta}{1-\gamma}  \\
&\implies \norm{\tilde{J}_{\bm{\pi}}-J_{\bm{\pi}}}_{\infty}\leq \frac{\gamma \alpha \beta}{(1-\gamma)^2}.\qed
\end{align*}
\section{Example: Observation Mapping}\label{app:observation_example}
Consider the flow control use case (Fig.~\ref{fig:systems}.a). Assume that the attacker attempts a brute-force SSH login against a server in the target system. The mapping from raw events in the digital twin to observations that are input to the security strategy proceeds in three stages.

\paragraph{Stage 1 - raw log events} The attacker launches an SSH brute-force attack against a server with IP 172.31.2.10. This produces raw log entries on the target host, e.g., in \texttt{/var/log/auth.log}:
\small
\begin{verbatim}
Mar 12 09:14:01 srv2 sshd[4821]:
Failed password for root from
172.31.1.42 port 49822 ssh2
Mar 12 09:14:02 srv2 sshd[4821]:
Failed password for root from
172.31.1.42 port 49822 ssh2
...
\end{verbatim}
\normalsize

Simultaneously, the network traffic is captured by the monitoring agent, which records packet-level statistics (e.g., flow byte counts and connection durations).

\paragraph{Stage 2 - Detection alerts}
The Snort IDS, running on the same network segment, matches the traffic against its ruleset and produces alerts:
\small
\begin{verbatim}
[**] [1:2001219:20] ET SCAN Potential
SSH Brute Force [**]
[Priority: 2] 03/12-09:14:05.003
172.31.1.42:49822 -> 172.31.2.10:22
[Classification: Attempted Information Leak]

[**] [1:2003068:7] ET SCAN Potential
SSH Login Attempt [**]
[Priority: 2] 03/12-09:14:18.112
172.31.1.42:49830 -> 172.31.2.10:22
[Classification: Misc Attack]
\end{verbatim}
\normalsize
The monitoring agent on the host pushes both the raw metrics and the IDS alerts to the Kafka event bus. A data pipeline then consumes  these events, aggregates them over a fixed monitoring interval (e.g., 15 seconds), and writes the results to the metastore.

\paragraph{Stage 3 - Mapping to observations}
At the end of each monitoring interval, the aggregated data is transformed into a numerical observation that serves as input to the security strategy learned through reinforcement learning. For the flow control POMDP, the observation at time step t is:
\begin{align*}
  &o_t=(\text{number of severe alerts}, \text{number of warning alerts}, \\
  &\quad\quad\quad\text{number of login events}).
\end{align*}
For the example above, supposing the monitoring interval covers the attack window, one realization might be:
\begin{align*}
  o_t=(1, 9, 4).
\end{align*}
This observation is then used to update the defender's belief state $b_t$, which is input to the security strategy $\pi_{\mathrm{D}}$.
\section{Example: Digital Twin Configuration}\label{app:digital_twin_config}
A complete configuration of a digital twin of the system in Fig.~\ref{fig:systems}.a is available at \url{https://github.com/Kim-Hammar/csle/blob/master/emulation-system/envs/090/level_9/config.py}.
\section{Artifact Appendix}\label{app:artifact_appendix}

\subsection{Abstract}

All results presented in this paper are fully reproducible using open-source software and data. To enable independent verification of our results and encourage future research, we release a complete set of artifacts that allow the community to build upon our work without additional engineering effort.

Specifically, we provide the following artifacts:
\begin{itemize}
  \item The source code of CSLE, our reinforcement learning platform for autonomous security management. Available at: \url{https://github.com/Kim-Hammar/csle}.
  \item Pre-built Docker images for creating digital twins. Available at: \url{https://hub.docker.com/u/kimham}.
  \item A video that demonstrates how to install CSLE. Available at: \url{https://www.youtube.com/watch?v=l_g3sRJwwhc}.
  \item A video that demonstrates the web interface of CSLE. Available at: \url{https://www.youtube.com/watch?v=iE2KPmtIs2A}.
  \item A dataset of attack statistics generated with CSLE. Available at: \url{https://github.com/Kim-Hammar/csle/releases/download/v0.4.0/statistics_dataset_14_nov_22_json.zip}.
  \item Platform documentation. Available at: \url{https://github.com/Kim-Hammar/csle/blob/master/releases} (PDF version) and \url{https://kim-hammar.github.io/csle/} (web version).
\end{itemize}    

\subsection{Artifact check-list (meta-information)}

{\small
\begin{itemize}
  \item {\bf Dataset: }\url{https://github.com/Kim-Hammar/csle/releases/download/v0.4.0/statistics_dataset_14_nov_22_json.zip}.
  \item {\bf Publicly available?: } Yes, \url{https://github.com/Kim-Hammar/csle} and \url{https://hub.docker.com/u/kimham}.
  \item {\bf Code licenses (if publicly available)?: } CC-BY-SA 4.0.
  \item {\bf Data licenses (if publicly available)?: } CC-BY-SA 4.0.
  \item {\bf Archived: } Yes, \url{https://doi.org/10.5281/zenodo.18869003}.
\end{itemize}
}
\subsection{Description}

\subsubsection{How delivered}

The code, data, and documentation are available on GitHub: \url{https://github.com/Kim-Hammar/csle}. The video demonstrations are available on YouTube: \url{https://www.youtube.com/watch?v=l_g3sRJwwhc} and \url{https://www.youtube.com/watch?v=iE2KPmtIs2A}. The Docker images are available on DockerHub: \url{https://hub.docker.com/u/kimham}.

\subsubsection{Hardware dependencies}
Minimum hardware requirements to run CSLE are: 16GB memory (RAM), 1 CPU, and 50GB disk space.

\subsubsection{Software dependencies}
Docker, Python, PostgreSQL, Prometheus, cAdvisor, Grafana, Node, and NPM. A comprehensive list of software dependencies is available in the platform documentation: \url{https://github.com/Kim-Hammar/csle/blob/master/releases}.

\subsubsection{Datasets}

A dataset of attack statistics generated with CSLE is available at: \url{https://github.com/Kim-Hammar/csle/releases/download/v0.4.0/statistics_dataset_14_nov_22_json.zip}.

\subsection{Installation}

The installation is automated through the Ansible playbooks available at: \url{https://github.com/Kim-Hammar/csle/tree/master/ansible}.

\subsection{Evaluation and expected result}

To validate the availability of our artifacts, perform the following steps: (\textit{i}) access the code, documentation, and data at \url{https://github.com/Kim-Hammar/csle}; (\textit{ii}) access the video demonstrations at \url{https://www.youtube.com/watch?v=l_g3sRJwwhc} and \url{https://www.youtube.com/watch?v=iE2KPmtIs2A}; and (\textit{iii}) access the Docker images at \url{https://hub.docker.com/u/kimham}.
\end{document}